\documentclass[aps,twocolumn,superscriptaddress,prb]{revtex4-1}

\usepackage[final]{graphicx}

\usepackage{amssymb}
\usepackage{cancel}
\usepackage{amsmath}
\usepackage{afterpage}
\usepackage{color}
\usepackage{amsthm}
\usepackage{hyperref}

\newtheorem{theorem}{Theorem}
\newtheorem{remark}{Remark}[theorem]
\newtheorem{lemma}{Lemma}

\begin{document}

\title{Superconvergence of Topological Entropy in the Symbolic Dynamics of Substitution Sequences}

\author{Leon Zaporski}
 \email{leon.zaporski@maths.ox.ac.uk}
\affiliation{Mathematical Institute, University of Oxford, Andrew Wiles Building, Radcliffe Observatory Quarter, Woodstock Road, Oxford, OX2 6GG, UK}
\author{Felix Flicker}
 \email{flicker@physics.org}
\affiliation{Rudolph Peierls Centre for Theoretical Physics, University of Oxford, Department of Physics, Clarendon Laboratory, Parks Road, Oxford, OX1 3PU, UK}

\date{\today}

\begin{abstract}
We consider infinite sequences of superstable orbits (cascades) generated by systematic substitutions of letters in the symbolic dynamics of one-dimensional nonlinear systems in the logistic map universality class. We identify the conditions under which the topological entropy of successive words converges as a double exponential onto the accumulation point, and find the convergence rates analytically for selected cascades. Numerical tests of the convergence of the control parameter reveal a tendency to quantitatively universal double-exponential convergence. Taking a specific physical example, we consider cascades of stable orbits described by symbolic sequences with the symmetries of quasilattices. We show that all quasilattices can be realised as stable trajectories in nonlinear dynamical systems, extending previous results in which two were identified.
\end{abstract}

\maketitle

%
\section{Introduction}
\label{sec:intro}
%

\begin{figure*}[t!]
\begin{minipage}[c]{\textwidth}
\includegraphics[width=\linewidth]{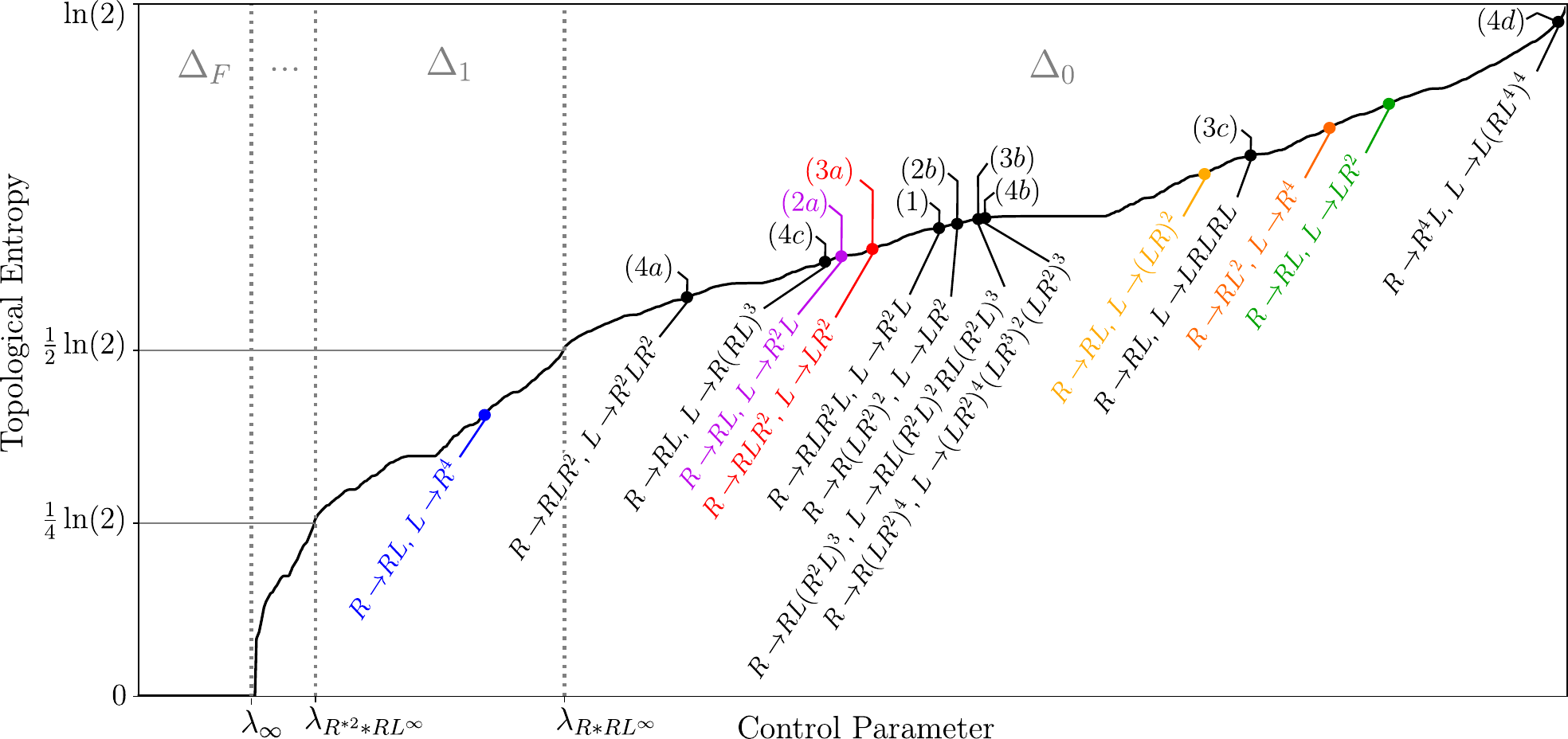}
\end{minipage}\hfill 
\caption{The topological entropy as a function of control parameter $\lambda$ in the Logistic map (Eq.~\eqref{eq:logistic}). The characteristic multifractal shape is known as a Devil's Staircase~\cite{PhysRevE.51.1983}. The plot was generated by finding the smallest positive zeros of the truncated kneading determinants numerically for itineraries of $f_\lambda(x)=1-\lambda x^2$ corresponding to $\lambda \in [1.35,2.00]$, uniformly sampled (see Section~\ref{sec:topological_entropy}). The point $\lambda_\infty$ corresponds to the infinite word $R^{*\infty}$ (notation defined in Section \ref{subsec:Fibonacci}). The point $\lambda_{R*RL^\infty}$ marks the lower bound of the interval $\Delta_0$, and the upper bound of the interval $\Delta_1$. The point $\lambda_{R^{*2}*RL^\infty}$ marks the lower bound of the interval $\Delta_1$ and the upper bound of the interval $\Delta_2$. The accumulation points of several substitution sequences are indicated, with the substitutions listed below the curve. Where colours are used, they are consistent between figures. In the cases that accumulation points correspond to generalised time quasilattices, the Boyle-Steinhardt class is indicated above the curve (see Section~\ref{sec:gTCs})~\cite{BoyleSteinhardt16}.
}\label{fig:Staircase}
\end{figure*} 

Chaos theory, governing systems featuring a lack of predictability under deterministic dynamical evolution, has found applications as varied as climate science, communications, economics, and fluid mechanics, to name a few~\cite{climate,cuomo,econ,knobloch_proctor_1981,PhysRevA.36.2862,Li2016}. The key to this wide applicability lies in \emph{universality}: a quantitative and qualitative similarity across seemingly disparate chaotic models~\cite{strogatz:2000,ChaosBook,AlligoodEA,KNOBLOCH1981439,Feigenbaum80}. 

One of the clearest demonstrations of universality is provided by \emph{symbolic dynamics}, the study of a system's dynamical behaviour when coarse-grained into discrete regions labeled by different symbols~\cite{ColletEckmann80,deBruijn81,Grassberger88,bailinhao}. Many properties of chaotic systems can be understood entirely in terms of the sequences of symbols (regions visited) in this manner~\cite{Isola1990}. One example is the `universal sequence' in which periodic windows develop when chaos is reached via a \emph{period doubling cascade}: an infinite sequence of period-doubling bifurcations~\cite{METROPOLIS197325,Sharkovskii64,strogatz:2000,ChaosBook,AlligoodEA,KNOBLOCH1981439,Feigenbaum80}. At each bifurcation, an initially stable periodic trajectory (orbit) becomes unstable, while a new trajectory of twice the period stabilises. Many important results concerning the period doubling cascade were established using the tools of symbolic dynamics~\cite{ColletEckmann80,deBruijn81,Grassberger88,bailinhao}. Period doubling continues to be of importance to cutting edge research: recent experiments established the existence of (discrete) \emph{time crystals}, which spontaneously break the symmetry of a periodic driving by returning a robust period-doubled response, made rigid to perturbations and finite temperature by the local interactions of many degrees of freedom~\cite{YaoEA17,ElseEA16,KhemaniEA16,ChoiEA17,ZhangEA17}. 

The present work is motivated in part by recent results establishing that periodically-driven nonlinear systems can feature not just period-doubled responses, but robust responses with the symmetries of one-dimensional \emph{quasilattices}~\cite{Fli18}. Well known in the context of crystallography, quasilattices are aperiodic long-range-ordered tilings, comprised of two or more unit cells, which possess a discrete scale invariance but which lack the discrete translational invariance of periodic crystal lattices~\cite{PhysRevB.34.617}. Mathematically they can be generated by repeatedly substituting tiles in a prescribed manner~\cite{Senechal,Janot}. In the context of dynamical systems, unit cells are replaced by trajectories with periods of fixed duration set, for example, by a periodic driving; the system spontaneously breaks the discrete time-translation symmetry of this driving by returning an aperiodic sequence of periods of two different durations. The sequence of cells can be generated mathematically by repeated application of substitutions to the symbolic dynamics governing the system~\cite{Fli18}. 

In the present work we consider a more general class of substitution rules applied to the symbolic dynamics of nonlinear systems. Given a periodic orbit described by an itinerary of coarse-grained regions visited in a dynamical system, an infinite sequence of new itineraries is generated by repeatedly substituting the symbols corresponding to the regions. Quasilattice substitution rules fall within the set we consider, and, by considering a simple generalisation of the basic quasilattice concept, we find that we are able to identify aperiodic orbits corresponding to all physically relevant quasilattices, extending previous results identifying two cases. Generalizing further we consider a set of substitutions additionally covering, for example, the period-doubling cascade~\cite{bailinhao}. The specific question we address is whether the complexity of the sequences in these generalised cascades shows any universal behaviour analogous to that shown in the period-doubling case.  

Using the notion of \emph{topological entropy} to quantify the complexity of symbolic sequences, we are able to make a number of precise statements about the development of complexity upon flowing down the cascades. Whereas the topological entropy is zero for all sequences in the period-doubling cascade, for other substitution sequences it increases monotonically~\cite{DH85,MThu,dMvS93,Dou95,Tsu00}. We find that the topological entropy of the wide class of substitution sequences we consider converges as a double exponential onto its accumulation point. We find the convergence rates analytically for some simple cascades, and outline the procedure for deriving the result for arbitrary cascades, given certain general assumptions. By numerically investigating the corresponding control parameters for specific unimodal maps, we further identify a second universal double-exponential convergence. Together, these results suggest further universal properties of the behaviour of these substitution sequences. In particular, the intervals within which admissible words have lengths greater than $n$ shrink geometrically as $n$ increases, in both the topological entropy (with geometric ratio $2$) and control parameter (with a system-dependent geometric ratio, previously identified in Ref.~\onlinecite{PhysRevLett.47.975}). 

This paper proceeds as follows. In Section~\ref{subsec:symbolicDynamics} we provide some relevant background to the field of symbolic dynamics, and define the conventions used in the paper. In Section~\ref{subsec:ressymdyn} we define the general class of substitutions we consider, and in Section~\ref{subsec:summarya} we collect a summary of the mathematical results presented throughout the remaining paper. In Section~\ref{sec:gTCs} we motivate the set of substitutions we consider by focussing on a specific subset of physical relevance, featuring the symmetries of quasilattices. We demonstrate that, under a simple generalisation, all one-dimensional quasilattices can appear as stable orbits in nonlinear dynamical systems. In Section~\ref{sec:topological_entropy} we return to the wider class of substitutions, and present results on the convergence of the topological entropy. In Section~\ref{sec:superconvergence} we demonstrate the superconvergence of the control parameter onto its accumulation point, in the same class of substitution sequences. We provide concluding remarks in Section~\ref{sec:Conclusions}.

%
\section{Background}
\label{sec:bacc}
%
 
%
\subsection{Symbolic Dynamics}
\label{subsec:symbolicDynamics}

One of the simplest systems to exhibit chaos is the one-parameter discrete-time \textit{logistic map}~\cite{mayr}:
\begin{equation}\label{eq:logistic}
x_{n+1}=\lambda x_n(1-x_n)
\end{equation}
defined on the interval $x_n\in\left[0,1\right]$. The dynamics become chaotic above a critical value $\lambda_0=3.56996..$. Chaos is reached through a cascade of supercritical pitchfork bifurcations at $\lambda<\lambda_0$, each of which doubles the period of the stable orbit~\cite{GiglioEA81}. The values of the control parameter $\lambda$ converge geometrically (with a ratio given by the now-famous Feigenbaum Constant $\delta=4.669\ldots$) to the accumulation point~\cite{Feig78}. This mathematical prediction was verified in various physical systems, including Rayleigh-B\'{e}nard convection and the Belousov Zhabotinsky reaction~\cite{RevModPhys.57.617,yahata,BZ}. 

The chaotic regimes of the continuous set of functions $f_\lambda$ are interspersed with \textit{periodic windows} of a finite range of $\lambda$, for which the system's behaviour converges onto a stable periodic orbit. Periodic windows were found to appear in a universal order for all unimodal maps~\cite{METROPOLIS197325,HaoEA83}. In order to establish the symbolic dynamics of map $f_\lambda$ we can assign the labels $L$, $R$ or $C$ to consecutive iterates of $f_\lambda$ if they are left, right, or central on the map (centred at the maximum, $x_c$), respectively~\cite{ColletEckmann80,Grassberger88}. In that sense, sequences of letters (\emph{words}) describe the itinerary (or kneading sequence) of points visited, coarse-grained to $L$, $R$, or $C$. Various tools have been developed to identify the words admissible as stable periodic orbits. In the following we focus on the \textit{generalised composition rules}, which systematically generate admissible words by a substitution process~\cite{deBruijn81}.

To take an example, the itineraries of positions visited by successive orbits in the period doubling cascade can be described by the substitutions
\begin{align}
R &\to RL\nonumber\\
L &\to RR
\label{eq:pd}
\end{align}
applied to the initial stable orbit $R$. At each stable orbit, the system exhibits periodic behaviour, generating an itinerary given by $\mathbf{K}=\bar{\mathbf{K}}^\infty$, where $\bar{\mathbf{K}}$ is a finite word, repeated infinitely to give the itinerary $\mathbf{K}$.

A word is defined to have an even \emph{parity} if it contains an even number of letters $R$, and an odd parity otherwise. A word $\bar{\mathbf{K}}$ with its last letter substituted with $C$, which we denote $\bar{\mathbf{K}}\rvert_C$, corresponds to an orbit containing the maximum at $x_c$, which ensures \emph{superstability} -- a faster-than-exponential convergence onto a (super)stable orbit~\cite{strogatz:2000}. The universal order of periodic windows coincides with the \emph{parity-lexicographic order} of words, defined through the relation `$\prec$' in the following way:
\begin{equation}
L \prec C \prec R
\end{equation}
and for two admissible words:
\begin{align}
&\bar{\mathbf{K}} =\bar{\mathbf{K}}^*\bar{\mathbf{L}}\nonumber\\ 
&\bar{\mathbf{K}}^\prime =\bar{\mathbf{K}}^*\bar{\mathbf{L}}^\prime
\end{align}
with $\bar{\mathbf{K}}^*$ their longest leftmost common substring (which may be the blank word, which has even parity). Denoting $l$ and $l^\prime$ the initial letters of $\bar{\mathbf{L}}$ and $\bar{\mathbf{L}}^\prime$:
\begin{align}
\bar{\mathbf{K}} \prec \bar{\mathbf{K}}^\prime \quad &\text{if } \bar{\mathbf{K}}^* \text{ has even parity, and } l \prec l^\prime \nonumber\\
\bar{\mathbf{K}}^\prime \prec \bar{\mathbf{K}} \quad &\text{if } \bar{\mathbf{K}}^* \text{ has odd parity, and } l \prec l^\prime.
\end{align}
A word is said to be \emph{maximal} if it is greater than or equal to all of its rightmost substrings (where equality holds only for the rightmost substring being the word itself)~\cite{METROPOLIS197325}. All maximal words are admissible~\cite{METROPOLIS197325,bailinhao}. 

Each admissible itinerary has a well-defined \textit{Topological Entropy}, an important measure of the dynamics' complexity~\cite{collet1983,muradkaki}. Consider the \textit{growth number} $s(n,\varepsilon)$, equal to the maximal possible number of distinguishable orbits after $n$ time steps, where two points belong to distinguishable orbits if they are separated by at least the distance $\varepsilon$. The topological entropy $h$ is then defined in the following manner~\cite{10.2307/1995565}:
\begin{equation}
h\equiv\lim_{\varepsilon \to 0} \lim_{n \to \infty} \frac{\text{ln}[s(n,\varepsilon)]}{n}.
\end{equation}
The topological entropy for discrete-time dynamical systems governed by unimodal maps is given by~\cite{Misiurewicz1980}:
\begin{equation}
h=\text{ln}\big(\lim_{i\to\infty}l^{1\slash i}_i\big)
\end{equation}
where $l_i$ is the number of laps (monotone intervals) of $f_\lambda^i$, \emph{i.e.} the $i^{\textrm{th}}$ iterate of $f_\lambda$. The higher the topological entropy, the more laps $\lim_{i\to\infty} f_\lambda^i$ has. In that sense, if we interpret $\lim_{i\to\infty} f_\lambda^i$ as a distribution of different initial values of $x$ acted on with $f_\lambda$ iteratively, infinitely many times, such a distribution is more complex for $\lambda$ corresponding to a higher topological entropy. 

The topological entropy depends only on the kneading sequence, not the form of the map $f_\lambda$~\cite{MThu}. Milnor and Thurston developed the concept of the \textit{kneading determinant}, providing a systematic way of calculating the topological entropy based only on the kneading sequence~\cite{MThu}. 

Topological entropy is preserved on intervals of $\lambda$ corresponding to period doubling cascades, as well as \textit{period} $K$\emph{-tupling} cascades -- sequences of orbits with periods increasing as $K^n$ ($K,n \in \mathbb{N}$), with itineraries generated by an alternative composition rule, the \textit{Derrida-Gervois-Pomeau star product} ($DGP*$)~\cite{PENG199443,Derrida1978}. All such cascades feature a geometric convergence of their topological entropies, characterised with universal parameters, which can be seen as a consequence of the associativity of the composition operators as well as the algebraic property of $DGP*$ leading to the conservation of topological entropy~\cite{0253-6102-3-3-283,XU20141505}. 

Plots of the evolution of the topological entropy with the control parameter take the form of a Devil's staircase, reproduced in Fig.~\ref{fig:Staircase}~\cite{PENG199443}. This features a multifractal structure (in the sense that the fractal dimensions of local portions of the graph are position dependent, and are elements of a continuous spectrum) and is non-decreasing, forming a piecewise-constant function, that is H\"older continuous~\cite{PhysRevE.51.1983,Tsu00,Isola1990,Bruin2009MonotonicityOE}. The interval marked $\Delta_F$ in the figure, $\lambda \in [0,\lambda_\infty)$ (of which only a small part is shown), features zero topological entropy. The intervals marked $\Delta_n$ ($n$ integer) exhibit geometric scaling by $0.5$ in the vertical direction and the Feigenbaum constant $\delta(R)=4.669\dots$ in the horizontal direction~\cite{PhysRevE.51.1983,collet1983}.

\subsection{Word Operations}
\label{subsec:ressymdyn}

Considering words written in an alphabet $\{L,R,C\}$, we define the following operations:
\begin{itemize}
\item $\bar{\mathbf{A}} \bar{\mathbf{B}}$ indicates the concatenation of words $\bar{\mathbf{A}}$ and $\bar{\mathbf{B}}$
\item $\lvert\bar{\mathbf{A}}\rvert$ returns the number of letters in $\bar{\mathbf{A}}$
\item $\lvert\bar{\mathbf{A}}\rvert_{R,L}$ returns the number of letters $R,L$ in $\bar{\mathbf{A}}$
\item $\bar{\mathbf{A}}\rvert_C$ substitutes the final letter of $\bar{\mathbf{A}}$ with the letter $C$.
\end{itemize}
Inverse words are defined as follows:
\begin{align}
\bar{\mathbf{A}}^{-1}\left(\bar{\mathbf{A}}\bar{\mathbf{B}}\right)=\bar{\mathbf{B}}\nonumber\\
\left(\bar{\mathbf{A}}\bar{\mathbf{B}}\right)\bar{\mathbf{B}}^{-1}=\bar{\mathbf{A}}.
\end{align}

We will be interested in substitution sequences of the form:
\begin{equation}\label{eq:subs_defs}
R \to  \bar{\mathbf{R}} \quad \text{and} \quad L \to  \bar{\mathbf{L}}.
\end{equation}
Assume that $\bar{\mathbf{R}}$ contains both letters $R$ and $L$ and starts with the letter $R$. We define an operator $\hat{S}(\bullet)$ that acts on word $ \bar{\mathbf{W}}$ by applying substitution rules to each of its letters. The operator $\hat{S}(\bullet)$ is distributive under concatenation. We prove an important theorem.

\begin{theorem}\label{secorder}
Substitution rules generating a cascade with initial word $\bar{\mathbf{W}}_{1}=R$ and $\bar{\mathbf{W}}_{2}=\bar{\mathbf{R}}$ can be restated as a second order linear recursive relation $\bar{\mathbf{W}}_{n+2}=g(\bar{\mathbf{W}}_{n},\bar{\mathbf{W}}_{n+1})$ under concatenation if $\bar{\mathbf{W}}_{3}=g(\bar{\mathbf{W}}_{1},\bar{\mathbf{W}}_{2})$. 
\end{theorem}
\begin{proof}
Such a relation can be written for the first three words. Assuming
\begin{equation*}
\bar{\mathbf{W}}_{n}=g(\bar{\mathbf{W}}_{n-2},\bar{\mathbf{W}}_{n-1})
\end{equation*} 
we have:
\begin{equation*}
\bar{\mathbf{W}}_{n+1}=\hat{S}(\bar{\mathbf{W}}_{n})=\hat{S}g(\bar{\mathbf{W}}_{n-2},\bar{\mathbf{W}}_{n-1}).
\end{equation*}
Since $\hat{S}(\bullet)$ is distributive, this gives
\begin{align}
\hat{S}g(\bar{\mathbf{W}}_{n-2},\bar{\mathbf{W}}_{n-1})=g(\hat{S}\bar{\mathbf{W}}_{n-2},\hat{S}\bar{\mathbf{W}}_{n-1})=g(\bar{\mathbf{W}}_{n-1},\bar{\mathbf{W}}_{n})
\end{align}
and therefore:
\begin{equation*}
\bar{\mathbf{W}}_{n+1}=g(\bar{\mathbf{W}}_{n-1},\bar{\mathbf{W}}_{n}).
\end{equation*} 
By induction,
\begin{equation*}
\bar{\mathbf{W}}_{n+2}=g(\bar{\mathbf{W}}_{n},\bar{\mathbf{W}}_{n+1})
\end{equation*}
holds for all natural numbers $n$.
\end{proof}

Not all substitution rules satisfy the conditions of Theorem~\eqref{secorder}. Take for example the substitution rules $R\to RLL$ and $L\to RLR$ (the period-tripling cascade). The cascade develops as follows:
\begin{equation}
R \to RLL \to RLLRLRRLR \to \ldots
\end{equation}
The third word cannot be written in terms of the first and second, owing to the isolated letters $L$. 

The substitution can be characterised by a $2\times 2$ \emph{growth matrix}
\begin{align}
A=\left(\begin{array}{cc}
a & b\\
c & d
\end{array}\right)
\end{align}
with non-negative integer entries, which quantifies the growth in the numbers of each letter type:
\begin{align}
\left(\begin{array}{c}
\left|\bar{\mathbf{W}}_{n}\right|_R\\
\left|\bar{\mathbf{W}}_{n}\right|_L
\end{array}\right)\rightarrow\left(\begin{array}{cc}
a & b\\
c & d
\end{array}\right)\left(\begin{array}{c}
\left|\bar{\mathbf{W}}_n\right|_R\\
\left|\bar{\mathbf{W}}_n\right|_L
\end{array}\right)=\left(\begin{array}{c}
\left|\bar{\mathbf{W}}_{n+1}\right|_R\\
\left|\bar{\mathbf{W}}_{n+1}\right|_L
\end{array}\right)
\end{align}
where repeated applications of the matrix correspond to multiple iterations of the substitutions. The class of substitutions we consider can then be written as
\begin{align}
\bar{\mathbf{W}}_{n}&=\bar{\mathbf{W}}_{n-1}\mathcal{P}\left(\bar{\mathbf{W}}_{n-1}^{\text{tr}\left(A\right)-1}\bar{\mathbf{W}}_{n-2}^{-\det\left(A\right)}\right)
\end{align}
for $n>2$, with $\bar{\mathbf{W}}_{1}=R$, and $\bar{\mathbf{W}}_{2}$ a specified word. The symbol $\mathcal{P}$ indicates an unspecified permutation. The characteristic equation of the growth matrix $A$ is
\begin{align}
\lambda^{2}&=\text{tr}\left(A\right)\lambda-\det\left(A\right).
\label{eq:W_characteristic}
\end{align}
The solutions (eigenvalues of $A$) must be real, and either integer or quadratic irrational. The second case is returned to in detail in Section~\ref{sec:gTCs} when we consider the special case of quasilattices. The ratio of the components of the eigenvector associated to the largest eigenvalue gives the relative frequencies of the two cell types~\cite{BoyleSteinhardt16}.
Eq.~\eqref{eq:W_characteristic} can be seen as the $n\rightarrow\infty$ limit of the defining equation of some integer sequence $W_n$ given by
\begin{align}
W_{n}&=\text{tr}\left(A\right)W_{n-1}-\det\left(A\right)W_{n-2}
\end{align}
for $n>2$, $W_1=\left|\bar{\mathbf{W}}_1\right|=1$, and $W_2=\left|\bar{\mathbf{W}}_2\right|$. The ratio $W_n/W_{n-1}$ gives the best possible rational approximation, for denominators not larger than $W_{n-1}$, to the largest eigenvalue of the growth matrix, \emph{i.e.} the larger of the solutions to Eq.~\eqref{eq:W_characteristic}.

Taking as an illustrative example the case of the period-doubling substitutions of Eq.~\eqref{eq:pd}, the words $\bar{\mathbf{W}}_n$ can be generated by
\begin{align}
\bar{\mathbf{W}}_n=\bar{\mathbf{W}}_{n-1}\bar{\mathbf{W}}^2_{n-2}
\end{align}
for $n>2$, with $\bar{\mathbf{W}}_1=R$ and $\bar{\mathbf{W}}_2=RL$. The first few cases are:
\begin{align}
R\rightarrow RL\rightarrow RLR^2\rightarrow RLR^3LRL\rightarrow RLR^3LRLRLR^3LR^2\rightarrow\ldots
\end{align}
The growth matrix can be concisely expressed if letters are taken to combine under addition rather than concatenation; forgiving this (hopefully intuitive) abuse of notation it is given by
\begin{align}
\left(\begin{array}{c}
R\\
L
\end{array}\right)\rightarrow\left(\begin{array}{cc}
1 & 1\\
2 & 0
\end{array}\right)\left(\begin{array}{c}
R\\
L
\end{array}\right)=\left(\begin{array}{c}
R+L\\
R+R
\end{array}\right)
\end{align}
with characteristic equation
\begin{align}
\lambda^2=\lambda+2
\end{align}
which can be seen as the $n\rightarrow\infty$ limit of the integer sequence $W_n$ given by
\begin{align}
W_n=W_{n-1}+2W_{n-2}
\end{align}
for $n>2$, with $W_1=\left|\bar{\mathbf{W}}_1\right|=\left|R\right|=1$ and $W_2=\left|\bar{\mathbf{W}}_2\right|=\left|RL\right|=2$. Explicitly, the first few terms are 
\begin{align}
1,\,2,\,4,\,8,\,16,\,32,\,64,\ldots
\end{align}
\emph{i.e.} $W_n=2^{n-1}$. The ratios of successive terms are all equal to 2, the largest eigenvalue of the growth matrix (equivalently, the largest solution to the characteristic equation).

For substitutions to generate cascades of words admissible as the itinerary of a stable trajectory, they must obey the \emph{generalised composition rules}, defined by the following conditions on $\bar{\mathbf{R}}$ and $\bar{\mathbf{L}}$~\cite{deBruijn81,bailinhao}:
\begin{enumerate}
\item $\bar{\mathbf{R}}$ has odd parity and $\bar{\mathbf{L}}$ has even parity
\item $\bar{\mathbf{R}} \succ \bar{\mathbf{L}}$
\item $\bar{\mathbf{R}}\rvert_C$ is maximal
\item $\bar{\mathbf{R}}\bar{\mathbf{L}}\rvert_C$ is maximal
\item $\bar{\mathbf{R}}(\bar{\mathbf{L}})^\infty$ is maximal.
\end{enumerate}
A special class of cascades, generated by so-called $DGP*$ composition rules, features a constant topological entropy. These cascades can always be restated as a set of substitutions satisfying the generalised composition rules:
\begin{equation}
R\to \bar{\mathbf{W}} K \quad \text{and} \quad L\to \bar{\mathbf{W}}  \widetilde{K}
\end{equation} 
with $K\in \{L,R\}$ and:
\begin{equation}
 \widetilde{K}=\begin{cases}
 L \quad \text{if } K=R \\
 R \quad \text{if } K=L.
 \end{cases}
 \end{equation}
For cascades beginning from $R$, such substitution rules preserve topological entropy from the second word onwards. 

The cascade generated by substitution rules $R\to RLL$ and $L\to RLR$, equivalent to the $DGP*$ composition rules generating period tripling, cannot satisfy the conditions of Theorem~\eqref{secorder}. However, there are substitution rules equivalent to $DGP*$ composition rules, that do satisfy the conditions of Theorem~\eqref{secorder}; for example $R\to RL$ and $L\to RR$ (period doubling) or $R\to RLRR$ and $L \to RLRL$ (period quadrupling). 

Finally, while the focus of this work is substitution sequences of the form specified by Eq.~\eqref{eq:subs_defs}, there are cascades not expressible by substitution rules that our results also apply to (such as the cascade considered in Section~\ref{subsec:Fibonacci}). 

\subsection{Summary of Results}
\label{subsec:summarya}

We consider cascades of words corresponding to superstable orbits of dynamical systems governed by unimodal maps, generated by substitution rules obeying the generalised composition rules. Specialising to the case in which the substitutions generate aperiodic words with the symmetries of quasilattices, we prove that all physically relevant quasilattices can be realized as words describing stable aperiodic orbits in the Logistic map universality class, extending the known result that two cases were possible~\cite{Fli18}.

We prove that if the first three words of a cascade form a second-order recursive relation in terms of string concatenation, the composition rule itself can be restated as a second order recursive relation -- see Theorem~\eqref{secorder}. Such a composition rule gives rise to an algebraic relation between polynomials related to the Milnor-Thurston kneading determinants, for each three consecutive words from the cascade -- see Theorem~\eqref{recdets}.

If the topological entropies of the second and third words in the cascade differ, then no two words from the cascade have equal topological entropies. If, addtionally, the lengths of the words grow exponentially fast (a generic feature) then the topological entropy converges as a double exponential onto the accumulation point $h_\infty$ (which is characteristic for a given cascade). Both statements require one additional identity concerning the aforementioned algebraic relation between polynomials to be satisfied -- see Lemma~\eqref{zerosnoteq} and Theorem~\eqref{main}. 

The asymptotic form of convergence (or \textit{stair-climbing} -- see Fig.~\ref{fig:Staircase}) can be found analytically. We do so for three cascades generated by the following substitution rules: 
\begin{itemize}
\setlength\itemsep{-0.2em}
\item $R\to RL$ and $L \to R^2L$, the Pell cascade (Section~\ref{subsec:Pell}) converges to the accumulation point $h_\infty=0.4411..$ with: 
\begin{align}
\text{ln}\big[ \text{ln}\big(\frac{h_n-h_{n-1}}{h_{n+1}-h_n}\big)\big]\sim n\text{ln}(1+\sqrt2)+\text{ln}(\tfrac{h_\infty}{2\sqrt2})
\end{align}
\item $R \to RLR^2$ and $L \to LR^2$, the Clapeyron cascade (Section~\ref{subsec:Clapeyron}) converges to the accumulation point $h_\infty=0.4484..$ with:
\begin{align}
\text{ln}\big[ \text{ln}\big(\frac{h_n-h_{n-1}}{h_{n+1}-h_n}\big)\big]\sim n\text{ln}(2+\sqrt3)+\text{ln}(\tfrac{1+\sqrt3}{6+4\sqrt3}h_\infty)
\end{align}
\item $R \to RL$ and $L \to L$ converges (geometrically) to the accumulation point $h_\infty=\text{ln}(2)$ with:
\begin{align}
\frac{h_n-h_{n-1}}{h_{n+1}-h_n}\sim 2.
\end{align}
\end{itemize}
Moreover, we show that our results extend to a larger class of cascades (Section~\ref{subsec:Fibonacci}). 

The case of $R \to RL$ and $L \to L$ reveals an interesting property of the Devil's staircase of topological entropy shown in Fig.~\ref{fig:convergence_entropy}~\cite{PENG199443}. Since each word $RL^n$ is the last admissible word of length $n+1$ (in the sense that all admissible words existing for $\lambda > \lambda_{RL^n}$ have length greater than $n+1$), the widths of intervals of topological entropy corresponding to words of lengths greater than $n$ shrink geometrically fast along the $y$-axis~\cite{bailinhao}. This supplements an earlier result showing that the convergence of the control parameter for superstable orbits is  geometrical~\cite{PhysRevLett.47.975}.

Finally, we study the convergence of the control parameter $\lambda_n$ numerically in cascades generated by substitution rules not equivalent to any $DGP*$ composition rule. Such a convergence is faster than geometrical (hence the name \textit{superconvergence}). Universal superconvergence has been found before in dynamical systems governed by multimodal maps~\cite{XU20141505}. The superconvergence of the control parameter in the class of cascades we focus on has the form:
\begin{equation}\label{convcontpar}
\text{ln}[-\text{ln}(\Lambda_{n} / \Lambda_{n+1})] \to An+\text{const} \quad \text{as } n \to \infty
\end{equation} 
where
\begin{equation}
\Lambda_n \equiv \frac{\lambda_n-\lambda_{n+1}}{\lambda_{n+1}-\lambda_{n+2}}.
\end{equation}
Numerical analysis suggests that the constant $A$ is independent of the form of $f_\lambda$, and is a universal characteristic for a given cascade. For example: $A=0.31\ldots$ in the cascade generated by $R\to RL$ and $L \to R^2L$, and $A=0.82\ldots$ in the cascade generated by $R \to RLR^2$ and $L \to LR^2$. 

%
\section{Generalised Time Quasilattices}
\label{sec:gTCs}
%

The physical motivation for considering the general class of substitution sequences specified in Section~\ref{subsec:ressymdyn} is provided by a particular subset of of these substitutions which lead to orbits with the symmetries of \emph{quasilattices}. In this section we consider this subset of sequences and show that, with a small generalisation of the basic quasilattice concept, all physically relevant quasilattices can appear as aperiodic words describing stable orbits in dynamical systems. We term these (generalised) \emph{time quasilattices}, in-keeping with the nomenclature laid out in Ref.~\onlinecite{Fli18}.

Quasilattices are one-dimensional aperiodic sequences formed from two unit cells of different lengths~\cite{GrunbaumShephard,BoyleSteinhardt16}. The substitution rules in this case are known as `inflation rules'~\cite{Senechal,Janot}. The simplest example is given by the Fibonacci quasilattice, generated by the inflation rules
\begin{align}
R\rightarrow RL,\quad L\rightarrow R.
\label{eq:Fibonacci_def}
\end{align}
Starting from the initial symbol $R$ and applying the rules to each symbol in the word results in a sequence of `Fibonacci words':
\begin{align}
R\rightarrow RL\rightarrow RLR\rightarrow RLR^2L\rightarrow RLR^2LRLR\rightarrow\ldots
\label{Fib_words}
\end{align}
The Fibonacci quasilattice $\bar{\mathbf{F}}_\infty$ results after an infinite number of inflations of $R$. We label the $n^{\textrm{th}}$ Fibonacci word $\bar{\mathbf{F}}_n$. Each Fibonacci word is the concatenation of the previous two:
\begin{align}
\bar{\mathbf{F}}_n=\bar{\mathbf{F}}_{n-1}\bar{\mathbf{F}}_{n-2}
\end{align}
for $n>2$, with $\bar{\mathbf{F}}_1=R$, and $\bar{\mathbf{F}}_2=RL$. Considering the inflation rule under addition rather than concatenation gives the growth matrix:
\begin{align}
\left(\begin{array}{c}
R\\
L
\end{array}\right)\rightarrow\left(\begin{array}{cc}
1 & 1\\
1 & 0
\end{array}\right)\left(\begin{array}{c}
R\\
L
\end{array}\right)=\left(\begin{array}{c}
R+L\\
R
\end{array}\right).
\end{align}
The characteristic equation of this matrix,
\begin{align}
\lambda^2=\lambda+1,
\label{eq:characteristic_fib}
\end{align}
has solutions $\varphi$, $\varphi^{-1}$, with $\varphi=\left(1+\sqrt{5}\right)/2$ a quadratic irrational number known as the golden ratio. The length of the $n^{\textrm{th}}$ Fibonacci word $\left|\bar{\mathbf{F}}_n\right|$ is the $n^{\textrm{th}}$ Fibonacci number $F_n$:
\begin{align}
F_n=F_{n-1}+F_{n-2}
\label{eq:Fib_number_def}
\end{align}
for $n>2$ with $F_1=\left|\bar{\mathbf{F}}_1\right|=\left|R\right|=1$ and $F_2=\left|\bar{\mathbf{F}}_2\right|=\left|RL\right|=2$. Starting from $F_0=1$ the first few lengths are 
\begin{align}
1,\,1,\,2,\,3,\,5,\,8,\,13,\,21,\,34,\,55,\,\ldots
\end{align}
Successive applications of the matrix correspond to successive inflations of the original sequence; successive word lengths also grow as Fibonacci numbers $F_n$:
\begin{align}
\left(\begin{array}{cc}
1 & 1\\
1 & 0
\end{array}\right)^{n}=\left(\begin{array}{cc}
F_{n+1} & F_{n}\\
F_{n} & F_{n-1}
\end{array}\right).
\end{align}
Noting that the ratio of successive Fibonacci numbers gives the best rational approximation to $\varphi$ for a given size of denominator, and that
\begin{align}
\lim_{n\rightarrow\infty}\frac{F_n}{F_{n-1}}=\varphi,
\end{align}
Eq.~\eqref{eq:characteristic_fib} can be seen as the $n\rightarrow\infty$ limit of the defining relation of Fibonacci numbers, Eq.~\eqref{eq:Fib_number_def}.

Consider the general case class of substitutions in two-letter alphabets specified in Section~\ref{subsec:ressymdyn}. In each case, the growth of the words under the substitution is characterised by a $2\times 2$ growth matrix $A$ with non-negative integer entries. The eigenvalues are real and given by
\begin{align}
\lambda_\pm=\frac{\textrm{tr}\left(A\right)}{2}\pm\sqrt{\left(\frac{\textrm{tr}\left(A\right)}{2}\right)^2-\det{A}}.
\label{eq:characteristic}
\end{align}
If $\textrm{tr}\left(A\right)^2=4\det{A}$ they are integers. Otherwise, the larger eigenvalue is a quadratic irrational `Pisot-Vijayaraghavan' (PV) number: the largest root of an irreducible monic polynomial, all of whose Galois conjugates have modulus strictly less than one~\cite{GrunbaumShephard,Senechal,Janot,BoyleSteinhardt16}. In this case the monic polynomial is just the quadratic characteristic equation of the matrix, and the Galois conjugate is simply the smaller eigenvalue. The converse, that all quadratic irrational PV numbers can be written as the eigenvalues of $2\times 2$ matrices with non-negative integer entries, follows from the ability to specify the trace and determinant independently in Eq.~\eqref{eq:characteristic}. Explicitly, we can always find a $2\times 2$ matrix $A$ with non-negative integer entries such that the PV number is the largest solution to
\begin{align}
\lambda^2=\textrm{tr}\left(A\right)\lambda-\det{A}
\end{align}
where $\textrm{tr}\left(A\right)$ and $\det{A}$ are uniquely specified by the PV number itself.

The following three conditions are necessary and sufficient for the substitutions to correspond to quasilattice inflation rules~\cite{BombieriTaylor86,BoyleSteinhardt16}:
\begin{enumerate}
\item the growth matrix must be unimodular
\item there must be two spacings between each symbol
\item the largest eigenvalue of the growth matrix must be a PV number.
\end{enumerate}
Table~\ref{tab:PV} lists the first few quadratic irrational PV numbers, their characteristic equations, and an example matrix with this characteristic equation. In the cases where the determinant of the matrix is of unit magnitude, the numbers also correspond to quasilattice inflation rules. The Boyle-Steinhardt class of the physically-relevant cases, to be discussed shortly, is also given.

\noindent \begin{center}
\begin{table}
\noindent \begin{centering}
\begin{tabular}{|c|c|c|c|}
\hline 
PV number & equation & example & QL class\tabularnewline
\hline 
\hline 
$\frac{1+\sqrt{5}}{2}$ & $\lambda^{2}=\lambda+1$ & $\left(\begin{array}{cc}
1 & 1\\
1 & 0
\end{array}\right)$ & 1\tabularnewline
\hline 
$1+\sqrt{2}$ & $\lambda^{2}=2\lambda+1$ & $\left(\begin{array}{cc}
1 & 1\\
2 & 1
\end{array}\right)$ & 2\tabularnewline
\hline 
$\frac{3+\sqrt{5}}{2}$ & $\lambda^{2}=3\lambda-1$ & $\left(\begin{array}{cc}
2 & 1\\
1 & 1
\end{array}\right)$ & 1\tabularnewline
\hline 
$1+\sqrt{3}$ & $\lambda^{2}=2\lambda+2$ & $\left(\begin{array}{cc}
1 & 1\\
3 & 1
\end{array}\right)$ & -\tabularnewline
\hline 
$\frac{3+\sqrt{13}}{2}$ & $\lambda^{2}=3\lambda+1$ & $\left(\begin{array}{cc}
3 & 1\\
1 & 0
\end{array}\right)$ & -\tabularnewline
\hline 
$2+\sqrt{2}$ & $\lambda^{2}=4\lambda-2$ & $\left(\begin{array}{cc}
2 & 2\\
1 & 2
\end{array}\right)$ & -\tabularnewline
\hline 
$\frac{3+\sqrt{17}}{2}$ & $\lambda^{2}=3\lambda+2$ & $\left(\begin{array}{cc}
1 & 2\\
2 & 2
\end{array}\right)$ & -\tabularnewline
\hline 
$2+\sqrt{3}$ & $\lambda^{2}=4\lambda-1$ & $\left(\begin{array}{cc}
1 & 2\\
1 & 3
\end{array}\right)$ & 3\tabularnewline
\hline 
$\frac{3+\sqrt{21}}{2}$ & $\lambda^{2}=3\lambda+3$ & $\left(\begin{array}{cc}
1 & 5\\
1 & 2
\end{array}\right)$ & -\tabularnewline
\hline 
$2+\sqrt{5}$ & $\lambda^{2}=4\lambda+1$ & $\left(\begin{array}{cc}
3 & 1\\
4 & 1
\end{array}\right)$ & 4\tabularnewline
\hline 
\end{tabular}
\par\end{centering}
\caption{\label{tab:PV}The first few quadratic-irrational Pisot-Vijayaraghavan
(PV) numbers. Each is the root of a monic quadratic equation uniquely
specified by the number, listed in the second column. This equation
always corresponds to the characteristic equation of a $2\times2$
matrix with non-negative integer entries; an example is given in each
case in column three. The trace is given by the coefficient of the
linear term in the equation, and the determinant the negative of the
constant term. If the determinant has unit magnitude the number relates
to a quasilattice inflation rule. The Boyle-Steinhardt class of the
quasilattice is listed where applicable (see Table \ref{tab:quasilattices}). }
\end{table}
\par\end{center}

\begin{center}
\begin{table}
\begin{centering}
\begin{tabular}{|c|c|c|c|c|}
\hline 
Class & $A$ & power & $R\rightarrow\bar{\boldsymbol{R}}$ & $L\rightarrow\bar{\boldsymbol{L}}$\tabularnewline
\hline 
\hline 
1 & $\left(\begin{array}{cc}
1 & 1\\
1 & 0
\end{array}\right)$ & 3 & $RLR^{2}L$ & $R^{2}L$\tabularnewline
\hline 
\hline 
2a & $\left(\begin{array}{cc}
1 & 1\\
2 & 1
\end{array}\right)$ & 1 & $RL$ & $R^{2}L$\tabularnewline
\hline 
2b & $\left(\begin{array}{cc}
2 & 1\\
1 & 0
\end{array}\right)$ & 2 & $R\left(LR^{2}\right)^{2}$ & $LR^{2}$\tabularnewline
\hline 
\hline 
3a & $\left(\begin{array}{cc}
3 & 1\\
2 & 1
\end{array}\right)$ & 1 & $RLR^{2}$ & $LR^{2}$\tabularnewline
\hline 
3b & $\left(\begin{array}{cc}
2 & 1\\
3 & 2
\end{array}\right)$ & 2 & $RL\left(R^{2}L\right)^{3}$ & $RL\left(R^{2}L\right)^{2}RL\left(R^{2}L\right)^{3}$\tabularnewline
\hline 
3c & $\left(\begin{array}{cc}
1 & 1\\
2 & 3
\end{array}\right)$ & 1 & $RL$ & $LRLRL$\tabularnewline
\hline 
\hline 
4a & $\left(\begin{array}{cc}
3 & 1\\
4 & 1
\end{array}\right)$ & 1 & $RLR^{2}$ & $R^{2}LR^{2}$\tabularnewline
\hline 
4b & $\left(\begin{array}{cc}
2 & 1\\
5 & 2
\end{array}\right)$ & 2 & $R\left(LR^{2}\right)^{4}$ & $\left(LR^{2}\right)^{4}\left(LR^{3}\right)^{2}\left(LR^{2}\right)^{3}$\tabularnewline
\hline 
4c & $\left(\begin{array}{cc}
1 & 1\\
4 & 3
\end{array}\right)$ & 1 & $RL$ & $R\left(RL\right)^{3}$\tabularnewline
\hline 
4d & $\left(\begin{array}{cc}
0 & 1\\
1 & 4
\end{array}\right)$ & 2 & $RL^{4}$ & $L\left(RL^{4}\right)^{4}$\tabularnewline
\hline 
\end{tabular}
\par\end{centering}
\caption{\label{tab:quasilattices}The Boyle-Steinhardt classification of the
ten physically-relevant quasilattices~\cite{BoyleSteinhardt16}. The first column lists the
class, and the second column lists the growth matrix for the simplest
substitution in that class, the largest eigenvalue of which is the
PV number listed in Table \ref{tab:PV} (consistent across the class). We consider the admissibility
of these quasilattice inflation rules as cascades of words describing
stable orbits in the logistic map universality class. The third column
states the minimum number of compound inflations to describe an admissible
substitution (the corresponding growth matrix will be the matrix in
column two, raised to this power). The final columns list the inflations
themselves. An infinite number of inflations applied to the word $R$
results in a (generalised) time quasilattice.
}
\end{table}
\par\end{center}

The first condition, $\left|\det{A}\right|=1$, implies that, since the growth matrix is an integer matrix, its inverse is also an integer matrix. The inflation (substitution) of any quasilattice sequence can therefore be undone with a well-defined deflation. This endows quasilattices with a discrete scale invariance~\cite{BoyleEA18}. This condition combines with the second condition, that the two cell types appear with two spacings, to imply that the two possible spacings differ by one (\emph{i.e.} $R$ can appear spaced by $n$ or $n+1$ $L$s, and $L$ can appear spaced by $m$ or $m+1$ $R$s). In the Fibonacci quasilattice, $R$ appears sandwiching either 0 or 1 $L$s, and $L$ appears sandwiching either 1 or 2 $R$s.
Note, however, that the growth matrix itself is not sufficient to guarantee the second condition. For example, the growth matrix
\begin{align}
\left(\begin{array}{cc}
1 & 1\\
2 & 1
\end{array}\right)
\end{align}
could correspond to the inflation rules
\begin{align}
R\rightarrow RL,\quad L\rightarrow R^2L
\end{align}
with the first few terms being
\begin{align}
R\rightarrow RL\rightarrow RLR^2L\rightarrow RLR^2LRLRLR^2L\rightarrow\ldots
\end{align}
or it could correspond to 
\begin{align}
R\rightarrow RL,\quad L\rightarrow LR^2
\end{align}
with the first few terms being
\begin{align}
R\rightarrow RL\rightarrow RL^2R^2\rightarrow RL^2R^2LR^3LRL\rightarrow\ldots
\end{align}
The first option is the Pell quasilattice, considered shortly (purple in Fig.~\ref{fig:Staircase}), with $\left\{0,1\right\}$ spaces between $R$s and $\left\{1,2\right\}$ spaces between $L$s. The second option (green in Fig.~\ref{fig:Staircase}) is not a quasilattice, as it features $\left\{0,1,2\right\}$ spaces between $L$s. It is therefore necessary to check the first few terms of the sequence to confirm the substitutions are generating a quasilattice.

The first and second conditions ensure an important alternative construction for the sequence. Shown in Fig.~\ref{fig:quasilattices}, an irrationally-sloped line is drawn intersecting a two-dimensional periodic lattice. Writing $R$ for each intersection with a vertical line of the lattice, and $L$ for each intersection with a horizontal line, the irrational slope (condition 3) ensures that the sequence of $R$s and $L$s is aperiodic. The figure shows the Fibonacci quasilattice being generated by a line with gradient $\varphi^{-1}$; shifting the intersecting line perpendicular to itself generates an uncountably infinite set of different quasilattices which are \emph{locally isomorphic}, meaning that every finite sequence of $R$s and $L$s appearing in one appears in all others~\cite{Senechal,Janot}. Note that shifting the line in this manner causes local $R\leftrightarrow L$ re-arrangements, which is the source of the requirement that the possible letter spacings differ only by one. While not a focus of the present work, this concept of quasilattices relating to higher-dimensional lattices is discussed at length in Refs.~\onlinecite{BoyleSteinhardt16,Fli18,FlickervanWezel15}. 

\begin{figure}[h!]
\includegraphics[width=\linewidth]{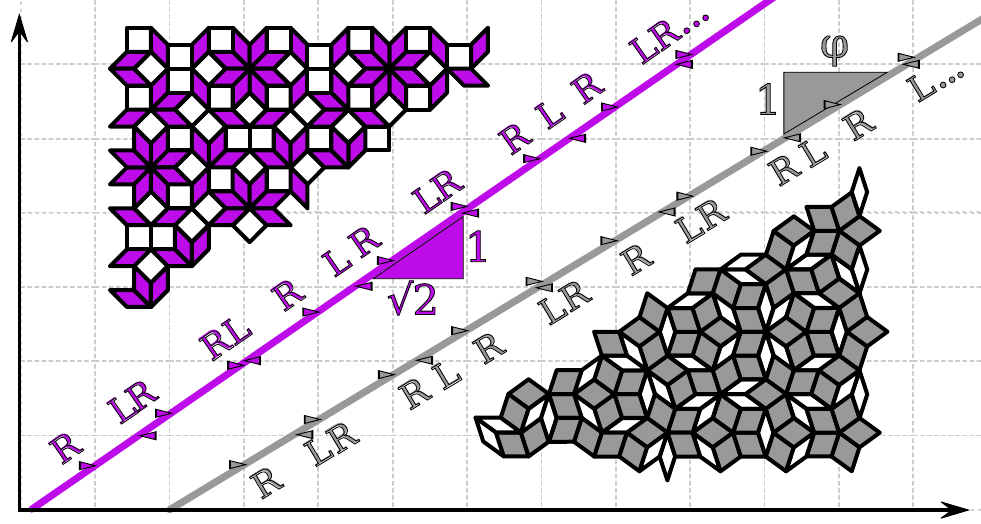}
\caption{The two cells of a quasilattice can be generated by drawing an irrationally-sloped line through a two-dimensional regular lattice; when the line cuts a vertical line, write $R$, and when it cuts a horizontal line, write $L$. The grey line shows the Fibonacci quasilattice being generated by a line with inverse slope $\varphi=\left(1+\sqrt{5}\right)/2$ the golden ratio; the purple line shows the Pell quasilattice being generated by a line with inverse slope related to $1+\sqrt{2}$ the silver ratio. Physically-relevant quasilattices have higher-dimensional counterparts in Penrose-like tilings with symmetries which can be realised by physical quasicrystals. The counterparts of the Fibonacci and Pell quasilattices are the Penrose tiling (grey and white, lower) and Ammann-Beenker tiling (purple and white, upper), respectively~\cite{BoyleSteinhardt16}.
\label{fig:quasilattices}
}
\end{figure}

The third condition, that the eigenvalues be non-integer (and therefore quadratic irrational PV numbers), is again necessary for the interpretation of the quasilattice sequence in terms of a cut through a higher-dimensional regular lattice. Without it, the aperiodic sequence would need to result from the intersections of a rationally-sloped line, which is impossible. In the case of the Fibonacci quasilattice, both the largest eigenvalue and the relative frequencies of the two cell types are given by $\varphi$, the smallest PV number. This condition also has important consequences for the diffraction pattern (Fourier transform) of the sequence: in particular, it forces the pattern to have sharp Bragg peaks, similar to a periodic system, but also a dense background, similar to a disordered system~\cite{BombieriTaylor86}. In fact, the Bragg spectrum of tilings with a discrete scale invariance has a non-trivial structure if and only if the scale factor of the tiling is a PV number~\cite{BombieriTaylor86,Moody}.

While there exist an infinite number of quasilattices satisfying the two criteria, and an infinite number of quadratic irrational PV numbers, it was identified by Boyle and Steinhardt in Refs.~\onlinecite{BoyleSteinhardt16,BoyleSteinhardt16B} that there are only ten physically relevant cases. The physical relevance derives from the fact that quasilattices in these classes have higher-dimensional counterparts whose symmetries correspond to those of physical \emph{quasicrystals}, states of matter intermediate between periodic crystals and disordered glasses~\cite{Levitov88}. The Fibonacci quasilattice has as its counterpart the two-dimensional Penrose tiling~\cite{Penrose74,Gardner}, governed by the PV number $\varphi^2$. This is shown in Fig.~\ref{fig:quasilattices}. Table~\ref{tab:quasilattices} lists each of the ten quasilattice classes along with its growth matrix. The determinant of each matrix has modulus one, and the largest eigenvalue is a PV number listed in Table \ref{tab:PV}. 

Starting from an orbit described by the word $R$, repeated application of the inflation rules will lead to a cascade of stable periodic orbits of increasing length. After an infinite number of substitutions, \emph{i.e.} at the accumulation point of the sequence, lies a stable orbit described by an aperiodic word: a \emph{time quasilattice}. Demonstrating the existence of such cascades is simply a matter of checking the quasilattice substitution rules against the generalised composition rules in Section~\ref{subsec:ressymdyn}. This was previously done in Ref.~\onlinecite{Fli18}, where two such cases were found. The first, class 2a in Table \ref{tab:quasilattices}, features the inflation rules
\begin{align}
R\rightarrow RL,\quad L\rightarrow R^2L.
\end{align}
Writing out the first few words
\begin{align}
R\rightarrow RL\rightarrow RLR^2L\rightarrow RLR^2LRLRLR^2L\rightarrow \ldots
\label{eq:Pell_cascade}
\end{align}
demonstrates that the symbols $R$ and $L$ appear with spaces $\left\{0,1\right\}$ and $\left\{1,2\right\}$, meeting the requirement that each symbol appears with two possible spaces with these spaces differing by one. The lengths of the words grow as Pell numbers $P_n$ (sequence $A000129$ in the Online Encyclopedia of Integer Sequences (OEIS))~\cite{sloane,Eulerlettre}. The words themselves are termed Pell words, with the $n^{\textrm{th}}$ word denoted $\bar{\mathbf{P}}_n$. The `Pell cascade' of Eq.~\eqref{eq:Pell_cascade} results in the Pell quasilattice $\bar{\mathbf{P}}_\infty$. As in the case of the Fibonacci words, the $n^\textrm{th}$ Pell word can also be generated as a concatenation of the previous two; in this case we have
\begin{align}
\bar{\mathbf{P}}_n=\bar{\mathbf{P}}_{n-1}\bar{\mathbf{P}}_{n-2}\bar{\mathbf{P}}_{n-1}
\end{align}
which mirrors the defining equation of the Pell numbers:
\begin{align}
P_n = 2P_{n-1} + P_{n-2}
\label{eq:Pell_num_def}
\end{align}
for $n>2$, $P_1=1$, $P_2=2$. Writing the Pell inflation rule under addition rather than concatenation results in the matrix equation
\begin{align}
\left(\begin{array}{c}
R\\
L
\end{array}\right)\rightarrow\left(\begin{array}{cc}
1 & 1\\
2 & 1
\end{array}\right)\left(\begin{array}{c}
R\\
L
\end{array}\right)=\left(\begin{array}{c}
R+L\\
2R+L
\end{array}\right).
\end{align}
The characteristic equation of this matrix again constitutes the $n\rightarrow\infty$ limit of the defining relation of the Pell numbers, Eq.~\eqref{eq:Pell_num_def}:
\begin{align}
\lambda^2=2\lambda+1.
\end{align}

A similar situation arises in class 3a. The inflation rules are 
\begin{align}
R\rightarrow RLR^2,\quad L\rightarrow LR^2
\end{align}
and writing out the first few cases starting from $R$ reveals that $R$ and $L$ appear with spaces $\left\{0,1\right\}$ and $\left\{2,3\right\}$ respectively. The growth matrix is
\begin{align}
\left(\begin{array}{cc}
3 & 1\\
2 & 1
\end{array}\right)\left(\begin{array}{c}
R\\
L
\end{array}\right)&=\left(\begin{array}{c}
R+L+R+R\\
L+R+R
\end{array}\right)
\end{align}
with characteristic equation
\begin{align}
\lambda^2=4\lambda-1.
\end{align}
This can again be seen as the $n\rightarrow\infty$ limit of the defining relation of an integer sequence. In this case it is the (modulus of the) Clapeyron numbers $C_n$ ($A125905$ in the OEIS)~\cite{sloane}:
\begin{align}
C_n=4C_{n-1}-C_{n-2}
\end{align}
for $n>2$ with $C_1=1$, $C_2=4$. In order to write the $n^{\textrm{th}}$ Clapeyron word $\bar{\mathbf{C}}_n$ in terms of the previous two we must make use of the inverse of a word, defined in Section~\ref{subsec:ressymdyn}:
\begin{align}
\bar{\mathbf{C}}_n=\bar{\mathbf{C}}_{n-1}\bar{\mathbf{C}}^{-1}_{n-2}\bar{\mathbf{C}}^{3}_{n-1}.
\end{align}

In order to establish whether the ten physical quasilattice classes in Table \ref{tab:quasilattices} can appear as words describing stable orbits in nonlinear dynamical systems, we coded an algorithm for testing arbitrary substitution rules against the generalised composition rules. The simplest case, the Fibonacci quasilattice, does not obey the generalised composition rules. For example, the Fibonacci substitutions of Eq.~\eqref{eq:Fibonacci_def} do not preserve the parity of the words under substitution (rule 1). However, three applications of the inflation rules \emph{do} preserve parity, and, in fact, obey all the generalised composition rules. This leads us to the conclusion that every third Fibonacci word can be realised as a stable periodic orbit in a nonlinear dynamical system in the Logistic map universality class. We term the process of carrying out $n$ inflations at each step $n^{\textrm{th}}$-order \emph{compound inflation}. The third-order compound inflation rules for the Fibonacci substitution rules are
\begin{align}
R\rightarrow RLR^2L,\quad L\rightarrow R^2L
\end{align}
giving the first few words
\begin{align}
R\rightarrow RLR^2L\rightarrow RLR^2LR^2LRLR^2LRLR^2LR^2L\rightarrow\ldots
\end{align}

For each of the ten classes in Table \ref{tab:quasilattices} we tested all inflation rules compatible with the growth matrix, including possible re-orderings of the substituted symbol sequences which still result in quasilattices. We also checked the cases with $L\leftrightarrow R$: the $L$ and $R$ labels are arbitrary in terms of the quasilattice unit cells, but can affect the admissibility of the sequences if $L$ is taken to correspond to the left of the dynamical system. If no solution was found, we considered a second-order compound inflation, corresponding to two powers of the substitution matrix $A$ (the second column in Table \ref{tab:quasilattices}), and further compound inflations until a solution was found. The words generated by compound inflation are a slight generalisation of the quasilattices considered in Refs.~\onlinecite{BoyleSteinhardt16} and \onlinecite{Fli18}. For example, the PV number governing the third-order Fibonacci compound inflation is $\varphi^3$ rather than $\varphi$ for the standard Fibonacci quasilattice. However, integer powers of a unimodular matrix are still unimodular, and integer powers of PV numbers are also PV numbers. Therefore the aperiodic sequences resulting from infinite numbers of compound inflations are still valid quasilattices. 

By this method, we were able to find time quasilattices for all ten quasilattice classes. This extends the results of Ref.~\onlinecite{Fli18} in which instances were found in classes 2a and 3a. Column three of Table \ref{tab:quasilattices} states the minimum order of compound inflation (or, equivalently, the minimum power to which the growth matrix $A$ must be raised) for the substitution to become admissible, and the final two columns list the substitutions themselves. 

Although we have not shown them in the table, it is also possible to find admissible quasilattice inflation rules with no higher-dimensional counterparts. Inspecting Table \ref{tab:PV} we see that the smallest PV number whose corresponding inflation matrix is unimodular, but which does not appear in Table \ref{tab:quasilattices}, is $\left(3+\sqrt{13}\right)/2$, with $\textrm{tr}\left(A\right)=3$ and $\det{A}=-1$. A substitution consistent with these conditions and with the generalised composition rules is 
\begin{align}
R&\rightarrow RL^4\left(RL^3\right)^2\nonumber\\
L&\rightarrow \left(LRL^3\right)^3\left(RL^3\right)^7
\end{align}
which we again verified using the word lifting technique. There are infinitely many such PV numbers with unimodular growth matrices. Based on our ability to find examples for all tested cases, it seems reasonable to expect that each can generate a time quasilattice, perhaps adjusting for compound inflation.

%
\section{Superconvergence of Topological Entropy}
\label{sec:topological_entropy}
%

Having focussed on a specific set of substitution sequences in Section~\ref{sec:gTCs} we now return to the general case in order to investigate the development of the words' complexities as we flow down the cascades under repeated substitutions. For word $\mathbf{K}$ (we will distinguish words corresponding to superstable orbits with overbars, $\bar{\mathbf{K}}$, in this section) the \emph{topological entropy} $h(\mathbf{K})$ can be calculated as $-\text{ln}[x^*(\mathbf{K})]$, where $x^*(\mathbf{K})$ is the smallest positive zero of the Milnor-Thurston kneading determinant $D_{\mathbf{K}}$~\cite{MThu}:
\begin{equation}
D_{\mathbf{K}}(x)=\sum^\infty_{n=0}\Theta^nx^n
\end{equation}
with the \textit{invariant co-ordinate} $\Theta^n$ defined as:
\begin{equation}
\begin{aligned}
\Theta^0 &\equiv 1\\
\Theta^n &\equiv \prod^n_{i=1} \epsilon_i \quad \text{for } n>0
\end{aligned}
\end{equation}
where, for a given kneading sequence $\mathbf{K}=K_1K_2\dots$:
\begin{equation}
\epsilon_i =\begin{cases}
     +1, & \text{if}\ K_i = L \\
     -1, & \text{if}\ K_i = R  \\
    \end{cases}
\end{equation}
and
\begin{equation}
\epsilon_{k+1}= \prod^k_{i=1} \epsilon_i, \quad \text{if } K_{k+1} = C.
\end{equation}

We reproduce an important identity~\cite{PENG199443}. Assume $\mathbf{K} = \bar{\mathbf{K}}^\infty$ with $\bar{\mathbf{K}}$ corresponding to a superstable orbit. We define a polynomial called the \textit{finite degree kneading determinant}:
\begin{equation}\label{kneaddet}
D_{\bar{\mathbf{K}}}(x)\equiv \sum^{\lvert  \bar{\mathbf{K}} \rvert -1}_{n=0}\Theta^nx^n
\end{equation}
recalling that $\lvert \bar{\mathbf{K}} \rvert$ is the total number of letters within word $\bar{\mathbf{K}}$. It follows that
\begin{align}
D_{\mathbf{K}}(x)= D_{\bar{\mathbf{K}}}(x) \big[1+x^{\lvert  \bar{\mathbf{K}} \rvert}+x^{2\lvert  \bar{\mathbf{K}} \rvert}+\dots \big] \nonumber\\ 
= \frac{1}{1-x^{\lvert  \bar{\mathbf{K}}\rvert}} D_{\bar{\mathbf{K}}}(x).
\end{align}
The first equality results from the fact that $\Theta^0=\Theta^{n\lvert  \bar{\mathbf{K}} \rvert}=1$ for $n=0,1,2,\dots$; the second equality holds for $x\in (-1,1)$. The spectrum of values of $h(\bar{\mathbf{K}}^\infty)$ indicates that $\tfrac{1}{2}\leq x^*(\bar{\mathbf{K}}^\infty)\leq1$, which simplifies the task of finding $x^*(\bar{\mathbf{K}}^\infty)$ to solving $D_{\bar{\mathbf{K}}}(x^*)=0$ for the smallest positive root~\cite{Tsu00}. 

Recalling that $\lvert \bar{\mathbf{W}} \rvert_R$ returns the number of letters $R$ within word $\bar{\mathbf{W}}$, words $\bar{\mathbf{W}}$ of odd parity will have $\lvert \bar{\mathbf{W}} \rvert_R$ odd and words of even parity will have $\lvert \bar{\mathbf{W}} \rvert_R$ even. 
\begin{theorem}\label{recdets} Take $\bar{\mathbf{E}}=\bar{\mathbf{F}}  \bar{\mathbf{G}}$. Then:
\begin{enumerate}
\item $D_{\bar{\mathbf{E}}}=D_{\bar{\mathbf{F}}}+(-1)^{\lvert \bar{\mathbf{F}}\rvert_R}x^{\lvert \bar{\mathbf{F}} \rvert}D_{\bar{\mathbf{G}}}$
\item $D_{\bar{\mathbf{F}}^{-1}\bar{\mathbf{E}}}=(-1)^{\lvert \bar{\mathbf{F}} \rvert_R }x^{-\lvert \bar{\mathbf{F}} \rvert}(D_{\bar{\mathbf{E}}}-D_{\bar{\mathbf{F}}})$
\item $D_{\bar{\mathbf{E}}\bar{\mathbf{G}}^{-1}}=D_{\bar{\mathbf{E}}}-(-1)^{\lvert \bar{\mathbf{E}} \bar{\mathbf{G}}^{-1} \rvert_R}x^{\lvert \bar{\mathbf{E}}\rvert -\lvert \bar{\mathbf{G}}\rvert}D_{\bar{\mathbf{G}}}$
\end{enumerate}
\end{theorem}
\begin{proof}
Identities 2 and 3 are just restatements of identity 1, which follows from the fact that $D_{\bar{\mathbf{F}}^{-1}\bar{\mathbf{E}}}=D_{\bar{\mathbf{G}}}$ and $D_{\bar{\mathbf{E}}\bar{\mathbf{G}}^{-1}}=D_{\bar{\mathbf{F}}}$. For clarity, we introduce some notation: $\Theta^k_{W}$ is an invariant co-ordinate for a substring $K_1K_2\dots K_k$ from $\bar{\mathbf{W}}$. From the definition of the invariant co-ordinate we have:
\begin{equation*}
\Theta^k_{E}=\begin{cases}
    \Theta^k_{F}, & \text{for}\ k<\lvert \bar{\mathbf{F}} \rvert \\
    (-1)^{\lvert \bar{\mathbf{F}}\rvert_R}\Theta^{k-\lvert \bar{\mathbf{F}} \rvert}_{G}, & \text{for}\ \lvert \bar{\mathbf{F}} \rvert \leq k < \lvert \bar{\mathbf{F}} \rvert + \lvert \bar{\mathbf{G}} \rvert. \\
     \end{cases}
\end{equation*}
From equation (\ref{kneaddet}) we have:
\begin{align}
D_{\bar{\mathbf{E}}}(x) &=\sum^{\lvert  \bar{\mathbf{E}} \rvert -1}_{n=0}\Theta_E^nx^n\nonumber\\
&=\sum^{\lvert  \bar{\mathbf{F}} \rvert -1}_{n=0}\Theta_E^nx^n+(-1)^{\lvert \bar{\mathbf{F}}\rvert_R}x^{\lvert \bar{\mathbf{F}} \rvert}\sum^{\lvert  \bar{\mathbf{G}} \rvert -1}_{n=0}\Theta_G^nx^n\nonumber\\
&=D_{\bar{\mathbf{F}}}(x)+(-1)^{\lvert \bar{\mathbf{F}}\rvert_R}x^{\lvert \bar{\mathbf{F}} \rvert}D_{\bar{\mathbf{G}}}(x)
\end{align}
\end{proof}
\begin{remark}\label{continuant}
Provided that cascade $\{\bar{\mathbf{W}}_n\}_{n\in\{0,1,2,3,..\}}$ is generated by a second order linear recursive relation 
\begin{equation}
\bar{\mathbf{W}}_n=g(\bar{\mathbf{W}}_{n-2},\bar{\mathbf{W}}_{n-1})
\end{equation}
under concatenation, Theorem~\eqref{recdets} ensures that:
\begin{equation}\label{eqncontinuant}
D_{\bar{\mathbf{W}}_n}(x)=a_n(x)D_{\bar{\mathbf{W}}_{n-1}}(x)+b_n(x)D_{\bar{\mathbf{W}}_{n-2}}(x)
\end{equation}
for $n \geq 3$ where $a_n(x)$ and $b_n(x)$ are polynomials. 
\end{remark}
\begin{remark}\label{bn}
This applies to all cascades satisfying the conditions of Theorem~\eqref{secorder}. In all of them we have:
\begin{equation}
\bar{\mathbf{W}}_n=g(\bar{\mathbf{W}}_{n-2},\bar{\mathbf{W}}_{n-1})=\bar{\mathbf{W}}_{n-1} g^*(\bar{\mathbf{W}}_{n-2},\bar{\mathbf{W}}_{n-1})
\end{equation}
ensured by the substitution rule $R\to\bar{\mathbf{R}}$, $R$ and $\bar{\mathbf{R}}$ being first two words in the cascade, and the fact that $\bar{\mathbf{R}}$ begins with the letter $R$ (a requirement of admissibility). This means that for such cascades $a_n(x)$ will have a constant term equal to $1$ and $b_n$ will have no constant term. 
\end{remark}

We introduce an important lemma:
\begin{lemma}\label{zerosnoteq}
Given a cascade satisfying the conditions from Remark~\eqref{continuant} with $b_n(\tau_{n-1}) \ne 0$ for $n \geq 3$:
\begin{equation*}
\tau_2 \ne \tau_1 \implies \tau_{n} \ne \tau_{n+1} \text{ for any } n\geq 2
\end{equation*}
where $\tau_n$ is the smallest positive zero of $D_{\bar{\mathbf{W}}_n}(x)$.
\end{lemma}
\begin{proof}
Let $\tau_M=\tau_{M-1}$ for some $M \geq 3$. Assume $\tau_N=\tau_{N-1}$ for some $N \leq M$. Equation~\eqref{eqncontinuant} for $n=N$ evaluated at $\tau_N$ implies
\begin{equation*}
0=b_N(\tau_N)D_{\bar{\mathbf{W}}_{N-2}}(\tau_N)
\end{equation*}
and since $b_n(\tau_{n-1}) \ne 0$ for $n \geq 3$, and $\tau_N=\tau_{N-1}$, it follows that $\tau_N$ is a zero of $D_{\bar{\mathbf{W}}_{N-2}}(x)$. Due to the monotonicity of topological entropy we have $\tau_{N-1}\leq \tau_{N-2}$. Therefore
\begin{equation*}
\tau_N=\tau_{N-1} \implies \tau_{N-1} = \tau_{N-2}.
\end{equation*}
By (backwards) induction we have $\tau_n=\tau_M$ for every $n\leq M$, and, in particular,
\begin{equation*}
[\text{There exists } M \geq 3 \text{ such that } \tau_M = \tau_{M-1}] \implies [\tau_2=\tau_1].
\end{equation*}
Therefore
\begin{equation*}
\tau_2 \ne \tau_1 \implies \tau_{n} \ne \tau_{n+1} \text{ for any } n\geq 2.
\end{equation*}
\end{proof}
Combining, we prove one of the main results of the paper:
\begin{theorem}\label{main}
Provided that a cascade $\{\bar{\mathbf{W}}_n\}_{n\in\{0,1,2,3,..\}}$ is generated by a second order linear recursive relation 
\begin{equation*}
\bar{\mathbf{W}}_n=g(\bar{\mathbf{W}}_{n-2},\bar{\mathbf{W}}_{n-1})=\bar{\mathbf{W}}_{n-1}  g^*(\bar{\mathbf{W}}_{n-2},\bar{\mathbf{W}}_{n-1})
\end{equation*}
under concatenation, and its finite degree kneading determinants $D_{\bar{\mathbf{W}}_n}(x)$ with smallest positive zeros $\tau_n$ satisfy
\begin{equation}\label{th3eqn}
D_{\bar{\mathbf{W}}_n}(x)=a_n(x)D_{\bar{\mathbf{W}}_{n-1}}(x)+b_n(x)D_{\bar{\mathbf{W}}_{n-2}}(x)
\end{equation}
with $b_n(\tau_{n-1})\neq0$ for every $n\geq 3$, the following is true:
\begin{equation}\label{implicationlimit}
[\tau_1 \neq \tau_2] \implies \Big[ \lim_{n\to \infty}\big(\frac{\tau_n-\tau_{n+1}}{\tau_{n-1}-\tau_n}\big)^{\lvert \bar{\mathbf{W}}_{n} \rvert ^{-1}}=\text{const.}\Big]
\end{equation}
with the value of the limit greater than zero. The form of the limit from the equation above does not hold for a few pathological cases of $\rvert \bar{\mathbf{W}}_{n} \lvert$ growing linearly with $n$.
\end{theorem}
\begin{proof}
Remark \eqref{bn} tells us that $b_n(x)$ in equation \eqref{th3eqn} has no constant terms. This means:
\begin{equation}\label{bnexplic}
b_n(x)=-x^{p\rvert \bar{\mathbf{W}}_{n} \lvert+q\rvert \bar{\mathbf{W}}_{n-1} \lvert}[1+\mathcal{O}(x^2)]
\end{equation}
where $p,q\in\mathbb{Z}$ and $\rvert \bar{\mathbf{W}}_{n} \lvert \leq p\rvert \bar{\mathbf{W}}_{n} \lvert+q\rvert \bar{\mathbf{W}}_{n-1} \lvert \leq \rvert \bar{\mathbf{W}}_{n+1} \lvert$.

Evaluating equation (\ref{th3eqn}) at $\tau_{n-1}$ gives
\begin{equation}\label{toexpand}
D_{\bar{\mathbf{W}}_n}(\tau_{n-1})=b_n(\tau_{n-1})D_{\bar{\mathbf{W}}_{n-2}}(\tau_{n-1})
\end{equation}
with $b_n(\tau_{n-1}) \neq 0$. Take $\delta_{n-1}=\tau_{n-1}-\tau_{n}$, and note it is greater than $0$ for every $n$ since $\tau_1 \neq \tau_2$ (see Lemma~\eqref{zerosnoteq}). Since the topological entropy $h(\bar{\mathbf{W}}_n^\infty)$ satisfies:
\begin{equation}
h_n \equiv h(\bar{\mathbf{W}}_n^\infty) = - \text{ln}(\tau_n)
\end{equation} 
and converges to a well-defined limit $h_\infty$ (as values of control parameters $\lambda_n$ within unimodal maps, that correspond to kneading sequences $\mathbf{K}_n=\bar{\mathbf{W}}_n^\infty$ converge to accumulation points at $\lambda_\infty$) we have:
\begin{equation}
\lim_{n\to\infty} \tau_n = \tau_\infty
\end{equation}
and therefore in the limit of large $n$, from the definition of the limit of the sequence $\{\tau_n\}_{n\in\mathcal{N}}$, we have $\delta_n \ll 1$.
This allows us to expand $D_{\bar{\mathbf{W}}_n}(x)$ from equation \eqref{toexpand} in powers of $\delta_{n-1}$ around $\tau_n$, and $D_{\bar{\mathbf{W}}_{n-2}}(x)$ from the same equation in powers of $\delta_{n-2}$ around $\tau_{n-2}$:
\begin{align}
\delta_{n-1}D^\prime_{\bar{\mathbf{W}}_n}(\tau_{n})+\mathcal{O}(\delta^2_{n-1})\nonumber\\
=b_n(\tau_{n-1})[-\delta_{n-2}D^\prime_{\bar{\mathbf{W}}_{n-2}}(\tau_{n-2})+\mathcal{O}(\delta^2_{n-2})].
\end{align}
Since $\delta_n \ll 1$ in the limit of large $n$, we have:
\begin{align}
\lim_{n\to\infty}\frac{\delta_{n-1}}{\delta_{n-2}} &=\lim_{n\to\infty}-b_n(\tau_{n-1})\frac{D^\prime_{\bar{\mathbf{W}}_{n-2}}(\tau_{n-2})}{D^\prime_{\bar{\mathbf{W}}_n}(\tau_{n})}\nonumber\\
&=\lim_{n\to\infty}-b_n(\tau_{n-1})\nonumber\\
&=\lim_{n\to\infty}\tau_{n-1}^{p\rvert \bar{\mathbf{W}}_{n} \lvert+q\rvert \bar{\mathbf{W}}_{n-1} \lvert}[1+\mathcal{O}(\tau_{n-1}^2)].
\end{align}
The second equality follows from the fact that, in the limit of large $n$, we have $D_{\bar{\mathbf{W}}_n}(x)\approx D_{\bar{\mathbf{W}}_{n-1}}(x)$ for $x\in[0,1)$ which can be seen directly from Remark \eqref{bn}; $a_n(x)$ has a constant term equal to 1, and all the other terms within $a_n(x)$ and $b_n(x)$ are proportional to $x^{\lvert \bar{\mathbf{W}}_n \rvert}$ and vanish as $\lvert \bar{\mathbf{W}}_n \rvert \to \infty$. In that limit, we therefore expect $D_{\bar{\mathbf{W}}_n}(x) \approx D_{\bar{\mathbf{W}}_{n-2}}(x)$. In the case that the first derivatives vanish at zero, we can always take higher order terms from the expansion; this would change the value of the limit, but not its constancy. The third equality follows from equation \eqref{bnexplic}. Finally, we end up with the expression:
\begin{align}
&\lim_{n\to\infty}\Big(\frac{\delta_{n-1}}{\delta_{n-2}}\Big)^{\rvert \bar{\mathbf{W}}_{n-1} \lvert^{-1}}\nonumber\\
&= \lim_{n\to\infty}\tau_{n-1}^{p\rvert \bar{\mathbf{W}}_{n} \lvert\rvert \bar{\mathbf{W}}_{n-1} \lvert^{-1} +q}[1+\mathcal{O}(\tau_{n-1}^2)]^{\rvert \bar{\mathbf{W}}_{n-1} \lvert^{-1}}\nonumber\\
&=\tau_{\infty}^{p \lim_{n\to\infty}(\rvert \bar{\mathbf{W}}_{n} \lvert\rvert \bar{\mathbf{W}}_{n-1} \lvert^{-1}) +q}=\text{const.}
\end{align}
or:
\begin{equation}\label{implimit}
\lim_{n\to \infty}\big(\frac{\tau_n-\tau_{n+1}}{\tau_{n-1}-\tau_n}\big)^{\lvert \bar{\mathbf{W}}_{n} \rvert ^{-1}}=\text{const}.
\end{equation}
\end{proof}
\begin{remark}\label{topentr}
Theorem \eqref{main} has a consequence for the convergence of the topological entropy. For cascades satisfying the conditions of Theorem~\eqref{main}, with $\tau_1 \neq \tau_2$, in the limit of large $n$ we have:
\end{remark}
\begin{align}
\frac{h_n-h_{n-1}}{h_{n+1}-h_n} &=\frac{- \text{ln}(\tau_n)+ \text{ln}(\tau_{n-1})}{{- \text{ln}(\tau_{n+1})+ \text{ln}(\tau_{n})}}\nonumber\\
&=\frac{- \text{ln}(\tau_n)+ \text{ln}(\tau_{n}+\delta_{n-1})}{- \text{ln}(\tau_n-\delta_n)+ \text{ln}(\tau_{n})}\nonumber\\
&=\frac{-\text{ln}(\tau_n)+\text{ln}(\tau_n)+\tau_n^{-1}\delta_{n-1} +\cancelto{0}{\mathcal{O}(\delta_{n-1}^2)}}{-\text{ln}(\tau_n)+\tau_n^{-1}\delta_{n}-\cancelto{0}{\mathcal{O}(\delta_{n}^2)}+\text{ln}(\tau_n) }\nonumber\\
&=\frac{\delta_{n-1}}{\delta_n}=\frac{\tau_{n-1}-\tau_n}{\tau_n-\tau_{n+1}}
\end{align}
\textit{which, in light of equation \eqref{implimit}, yields a double-exponential convergence of the topological entropy (provided $\rvert \bar{\mathbf{W}}_{n} \lvert$ grow exponentially fast, which is a generic feature of numbers defined by second-order linear recursive relations).} \medskip

We proceed to study the convergence of the topological entropy within specific cascades, and find its analytical asymptotic form. 

\subsection{The Pell Cascade}
\label{subsec:Pell}

Recall the substitution rule for the Pell cascade, $R \to RL$ and $L \to R^2L$ acting on $R$:
\begin{equation}\label{eq:Pell}
R \to RL \to RLR^2L \to RLR^2LRLRLR^2L \to \ldots
\end{equation}
The first three words from the cascade obey 
\begin{equation}
\bar{\mathbf{P}}_n=\bar{\mathbf{P}}_{n-1}\bar{\mathbf{P}}_{n-2}\bar{\mathbf{P}}_{n-1}.
\end{equation} 

Theorem \eqref{secorder} tells us that the Pell Cascade is generated by a second-order recursive relation:
\begin{equation}
\bar{\mathbf{P}}_n=\begin{cases}
     R, & \text{for}\ n=1 \\
     RL, & \text{for}\ n= 2  \\
     \bar{\mathbf{P}}_{n-1}\bar{\mathbf{P}}_{n-2}\bar{\mathbf{P}}_{n-1}, & \text{for}\ n > 2.  \\
    \end{cases}
\end{equation}
We proceed by writing down the finite degree kneading determinant for $\bar{\mathbf{P}}_{n+1}\rvert_C$ in terms of finite degree kneading determinants for $\bar{\mathbf{P}}_{n}\rvert_C$ and $\bar{\mathbf{P}}_{n-1}\rvert_C$ using the identities from Theorem~\eqref{recdets}. Recall that  $\lvert\bar{\mathbf{P}}_{n}\rvert \equiv P_n$, which corresponds to the $n^{\textrm{th}}$ Pell number.

We have:
\begin{align}\label{recrelpell}
D_{\bar{\mathbf{P}}_{n+1}}(x)=D_{\bar{\mathbf{P}}_{n}}(x)-x^{P_n}D_{\bar{\mathbf{P}}_{n-1}}(x)+x^{P_n+P_{n-1}}D_{\bar{\mathbf{P}}_{n}}(x).
\end{align}
The term $b_{n+1}(x)=-x^{P_n}$ has a single zero $x=0$, that does not coincide with the smallest positive zero of $D_{\bar{\mathbf{P}}_{n}}$, \emph{i.e.} $\tau_n$, for any $n$. Moreover $\tau_1 \neq \tau_2$ (see Fig.~\ref{fig:convergence_entropy}) so for $\delta_n=\tau_{n}-\tau_{n+1}$, due to Theorem~\eqref{main}:
\begin{equation}
\lim_{n \to \infty} \Big(\frac{\delta_n}{\delta_{n-1}}\Big)^{\tfrac{1}{P_n}}=\tau_\infty
\end{equation}
with $0<\tau_\infty<1$. The Pell number $P_n$ can be expressed by a closed formula:
\begin{equation}\label{pell_nos}
P_n=\frac{1}{2\sqrt2}\Big[\big(1+\sqrt2\big)^n-\big(1-\sqrt2\big)^n\Big].
\end{equation}
In the limit of large $n$ we have $P_n \sim \frac{1}{2\sqrt2}\big(1+\sqrt2\big)^n$ , which, due to Remark~\eqref{topentr} yields a double-exponential convergence of topological entropy. In the limit of large $n$ we have:
\begin{align}
\frac{h_n-h_{n-1}}{h_{n+1}-h_n} &\sim\tau_\infty^{ -\frac{1}{2\sqrt2}\big(1+\sqrt2\big)^n}\nonumber\\
&=\text{exp}\big\{\text{exp}[n\text{ln}(1+\sqrt2)+\text{ln}(\tfrac{h_\infty}{2\sqrt2})] \big\}.
\end{align}

Our numerical analysis indicates that $h_\infty=0.4411\ldots$ in the Pell Cascade. To test the theory we examine the asymptotic behaviour of the numerically generated sequence $\{h_n\}_{n\in\mathbb{N}}$:
\begin{equation}\label{defomega}
\Omega(n)\equiv\text{ln}\Big[\text{ln}\Big(\frac{h_n-h_{n-1}}{h_{n+1}-h_n}\Big)\Big]\to an+b.
\end{equation}
We used the Newton-Raphson method for finding consecutive values of $\tau_n$. We wrote a code exploiting this technique to collect the data, accurate to 800 decimal places. The calculated values of $\Omega(n)$ are plotted against $n$ in Fig.\ref{fig:convergence_entropy}. The sequence rapidly reaches its predicted asymptotic behaviour.

\subsection{The Clapeyron Cascade}
\label{subsec:Clapeyron}

The substitution rules $R \to RLR^2$ and $L \to LR^2$ applied to the initial word $R$ generate words with lengths given by the moduli of the Clapeyron numbers $C_n$. The first few words look as follows:
\begin{equation}\label{eq:Clapeyron}
R \to RLR^2 \to RLR^2LR^3LR^3LR^2 \to \ldots
\end{equation}
Owing to Theorem~\eqref{secorder}, the alternative composition rule generating consecutive words representing stable orbits is: 
\begin{equation}
\bar{\mathbf{C}}_{n+1}=\begin{cases}
     R, & \text{for}\ n=1 \\
     RLR^2, & \text{for}\ n=2  \\
     \bar{\mathbf{C}}_{n} \bar{\mathbf{C}}_{n-1}^{-1}  \bar{\mathbf{C}}^3_{n}, & \text{for}\ n > 2.  \\
    \end{cases}
\end{equation}
Using the identities from Theorem~\eqref{recdets} we find a recursive relation between finite degree kneading determinants for $\bar{\mathbf{C}}_{n+1}\rvert_C$, $\bar{\mathbf{C}}_{n}\rvert_C$ and $\bar{\mathbf{C}}_{n-1}\rvert_C$:
\begin{align}
D_{\bar{\mathbf{C}}_{n+1}}(x)= &D_{\bar{\mathbf{C}}_{n}}(x)-x^{C_n}D_{\bar{\mathbf{C}}_{n-1}^{-1} \bar{\mathbf{C}}_{n}}(x)+\nonumber\\
&-x^{2C_n-C_{n-1}}D_{\bar{\mathbf{C}}_{n}}(x)+x^{3C_n-C_{n-1}}D_{\bar{\mathbf{C}}_{n}}(x)
\end{align}
where
\begin{equation}
D_{\bar{\mathbf{C}}_{n}}(x)-D_{\bar{\mathbf{C}}_{n-1}}(x)=(-1)x^{C_{n-1}} D_{\bar{\mathbf{C}}_{n-1}^{-1}  \bar{\mathbf{C}}_{n}}(x).
\end{equation}
The factor $(-1)$ is a result of the odd parity of $\bar{\mathbf{C}}_{n-1}$. We introduced $D_{\bar{\mathbf{C}}_{n-1}^{-1}  \bar{\mathbf{C}}_{n}}(x)$ as a finite degree kneading determinant for $\bar{\mathbf{C}}_{n-1}^{-1}\bar{\mathbf{C}}_{n}\rvert_C$ (an inadmissible word) to provide an intermediate step in explaining the final form of the recursive relation:
\begin{align}
D_{\bar{\mathbf{C}}_{n+1}}(x)=&D_{\bar{\mathbf{C}}_{n}}(x)+x^{C_n-C_{n-1}}[D_{\bar{\mathbf{C}}_{n}}(x)-D_{\bar{\mathbf{C}}_{n-1}}(x)]\nonumber\\
&-x^{2C_n-C_{n-1}}D_{\bar{\mathbf{C}}_{n}}(x)+x^{3C_n-C_{n-1}}D_{\bar{\mathbf{C}}_{n}}(x).
\end{align}
Again, $b_{n+1}(x)=-x^{C_n-C_{n-1}}$ has single zero at $x=0$ which does not coincide with $\tau_n$ for any $n$. Moreover $\tau_1=\neq \tau_2$ (see Fig.~\ref{fig:convergence_entropy}).
Applying Theorem~\eqref{main} and Remark~\eqref{topentr} leads to:
\begin{equation}
\frac{h_n-h_{n-1}}{h_{n+1}-h_n}\sim\text{exp}\Big\{\text{exp}\Big[n\text{ln}\Big(2+\sqrt3\Big)+\text{ln}\Big(\tfrac{1+\sqrt3}{6+4\sqrt3}h_\infty\Big)\Big] \Big\}.
\end{equation}
Our numerical analysis shows that $h_\infty=0.4484\ldots$ As before, we test the above relation by comparing the double logarithm of the left hand side ($\equiv \Omega(n)$) with the asymptote of the form $an+b$. The asymptotic behaviour is reached very quickly (see Fig.~\ref{fig:convergence_entropy}).

\subsection{Miscellaneous Cases}
\label{subsec:Fibonacci}

There are cascades that cannot be generated by substitution rules, but which satisfy the conditions of Theorem~\eqref{main}. To take an example, consider the cascade generated by the second-order linear recursive relation
\begin{align}\label{eq:Fibonacci}
\bar{\mathbf{S}}_{3}\rvert_C &= RC\nonumber\\
\bar{\mathbf{S}}_{4}\rvert_C &= RLC\nonumber\\
\bar{\mathbf{S}}_{n+1}\rvert_C &= \bar{\mathbf{S}}_{n}\bar{\mathbf{S}}_{n-1}\rvert_C.
\end{align}
Word $\bar{\mathbf{S}}_{n}$ has length $\lvert \bar{\mathbf{S}}_{n} \rvert = F_n$ again corresponding to the $n^{\textrm{th}}$ Fibonacci number. Note, however, that the sequence of letters created by these substitutions does not match the sequence of cells in Fibonacci quasilattice (see Section~\ref{sec:gTCs}).

The proof of the admissibility of the words thus defined was given by K. Shibayama in Ref.~\onlinecite{shibayama}, and we term the resulting sequence of words the Shibayama cascade:
\begin{equation}
RC\to RLC \to RLLRC \to RLLRRRLC \to \dots
\end{equation}
which cannot be generated using substitution rules applied to $R$ and $L$. Analysis analogous to that in the previous sections leads to the result
\begin{equation}
\frac{h_n-h_{n-1}}{h_{n+1}-h_n}\sim\text{exp}\big\{\text{exp}[n\text{ln}(\tfrac{1+\sqrt5}{2})+\text{ln}(\tfrac{h_\infty}{\sqrt5})] \big\}
\end{equation}
with $h_\infty=0.5476\ldots$ (found numerically). Again, the asymptotic behaviour of $\Omega(n)$ was reached rapidly (see Fig.~\ref{fig:convergence_entropy}). 

Theorem \eqref{main} states that for the few cascades in which the lengths of the words grow linearly with $n$, the form of the limit is different. For all cascades generated by substitution rules with $R$ as the first word, the only cascade featuring linear growth is generated by $R \to RL$ and $L \to L$. However, the conditions of Theorem~\eqref{main} are satisfied in this case, and we end up with
\begin{equation}
\lim_{n\to\infty}\frac{h_n-h_{n-1}}{h_{n+1}-h_n}=2
\end{equation}
which corresponds to a geometric, rather than double-exponential, convergence. The accumulation point is $h_\infty=\text{ln}(2)$, so it reaches the top of the Devil's staircase of Fig.\ref{fig:Staircase}. Together with the geometric convergence of the control parameter for superstable orbits in such cascades (with a map-dependent, and therefore non-universal, geometric ratio) this leads to an interesting fact concerning the structure of the Devil's staircase of topological entropy~\cite{PhysRevLett.47.975}. Each word $RL^n$ is the last admissible word of length $n+1$ (in the sense that all admissible words existing for $\lambda > \lambda_{RL^n}$ have length greater than $n+1$), and points $(\lambda_{RL^n\rvert_C},h_{RL^n\rvert_C})$ converge geometrically fast along both the horizontal and vertical directions, up to the top of the staircase~\cite{bailinhao}. Since this cascade was studied thoroughly in Ref.~\onlinecite{PhysRevLett.47.975}, we omit any further discussion of it or any of its subcascades $\{ RL^{w_n}\}$.    

Other words with lengths growing linearly were qualitatively discussed before, using different techniques, in Ref.~\onlinecite{Isola1990}. However, we believe the quantitative uniformity in the convergence of the entropy is a new result of the present work.

Following the formula for topological entropy:
\begin{equation}\label{Eckmann}
h\big[(R^{*n}*\bar{\mathbf{M}})^\infty \big]= \frac{1}{2^n}h\big[\bar{\mathbf{M}}^\infty\big]
\end{equation}
derived by Collet, Crutchfield and Eckmann in Ref.~\onlinecite{collet1983}, where the action of the operator $R*$ is that of a single application of substitution rules $R\to RL$ and $L\to RR$ to a word (an example of a $DGP*$ composition rule), and $R^{*n}*=R*R*\dots*R*$ iterated $n$ times~\cite{Derrida1978}. The action of the $R^{*n}*$ operator on each word from a given cascade accumulating at a point in the interval $\Delta_m$ (see Fig.~\ref{fig:Staircase}) will transform such a cascade into another, accumulating at a point in the interval $\Delta_{m+n}$. In this new cascade, the entropy of the $k^{\textrm{th}}$ word, expressed in terms of the entropy of the $k^{\textrm{th}}$ word from the old cascade, is given by $h_k^\prime=2^{-n}h_k$, and the form of the expression for $\Omega(n)$ remains unchanged. This is due to the cancellation of the factors of $2^{-n}$ (see Eq.~\eqref{defomega}). Because of the perfect similarity of intervals $\Delta_n$, the convergence of the control parameters in the transformed cascades is expected to be quantitatively identical~\cite{PhysRevE.51.1983}. The transformed cascade is not equivalent to any cascade satisfying the conditions of Theorem~\eqref{secorder}, since the first word is not $R$. These results are indicated in Fig.~\ref{fig:Staircase}.

\begin{figure}[h!]
\includegraphics[width=\linewidth]{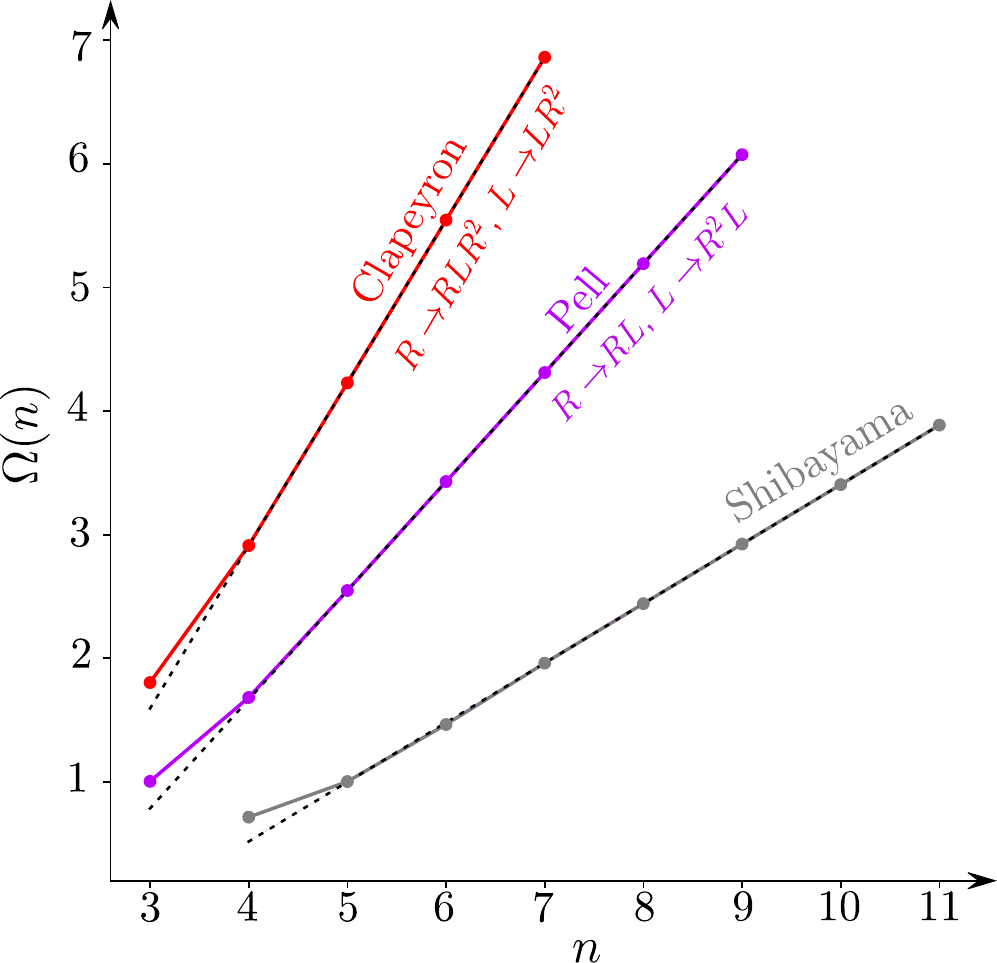}
\caption{The convergence of the topological entropy in the Pell and Clapeyron cascades, and the Shibayama cascade featuring words growing as Fibonacci numbers (Eqs.~\eqref{eq:Pell},\eqref{eq:Clapeyron}, and \eqref{eq:Fibonacci}, respectively). $\Omega(n)=\text{ln}[\text{ln}(\tfrac{h_n-h_{n-1}}{h_{n+1}-h_n})]$ is shown as a function of $n$, where $h_n$ corresponds to the topological entropy of the $n^{\textrm{th}}$ word from a given cascade. Dashed lines are asymptotes of the form $a_in+b_i$ with $a_P=\text{ln}(1+\sqrt2)$ and $b_P=\text{ln}(\tfrac{h_\infty}{2\sqrt2})$, $a_C=\text{ln}(2+\sqrt3)$ and $b_C=\text{ln}(\tfrac{1+\sqrt3}{6+4\sqrt3}h_\infty)$, $a_S=\text{ln}(\tfrac{1+\sqrt5}{2})$ and $b_S=\text{ln}(\tfrac{h_\infty}{\sqrt5})$, corresponding to the respective cascades.
\label{fig:convergence_entropy}
}
\end{figure}

%
\section{Superconvergence of Control Parameter -- Numerical Analysis}
\label{sec:superconvergence}
%

The results presented so far have concerned universal topological aspects of the kneading sequences. The values of the control parameter $\lambda$ in Eq.~\eqref{eq:logistic} will be specific to the Logistic map. Nevertheless, some statements can be made regarding the convergence of the control parameters $\lambda$ appearing within cascades of superstable orbits. With $\lambda_n$ the control parameter giving the $n^\textrm{th}$ orbit in a cascade, we define the geometric ratio of consecutive control parameters as
\begin{equation}
\Lambda_n \equiv \frac{\lambda_n-\lambda_{n+1}}{\lambda_{n+1}-\lambda_{n+2}}.
\end{equation}
It has been proven that $\Lambda_n \to \text{const}>0$ as $n \to \infty$ for period $n$-tupling cascades, generated by $DGP*$ composition rules~\cite{bailinhao}. Our numerical results show that cascades generated by non-$DGP*$ substitution rules feature a divergence of $\Lambda_n$ as $n \to \infty$, which corresponds to convergence faster than geometric (hence the name `superconvergence'). 

We used the technique of `word lifting', defined in Refs.~\onlinecite{bailinhao,Fli18}, to find the values of the control parameter $\lambda$ required to generate consecutive words in cascades. Calculations were performed maintaining up to 400 decimal places. We repeated numerical calculations for both the logistic map $f_\lambda(x)=1-\lambda x^2$, and another unimodal map in the same universality class, the sine map $f_\lambda(x)=\lambda \sin(\pi x)$. Based on the numerical data, we identified a double-exponential convergence, with
\begin{equation}\label{convcontpar}
\text{ln}[-\text{ln}(\Lambda_{n}/\Lambda_{n+1})] \to Dn+\text{const} \quad \text{as } n \to \infty.
\end{equation} 
The value of the parameter $D$ is identical, to numerical precision, for both the sine map and the logistic map, suggesting it is a universal characteristic of a given cascade. For example, $D=0.31..$ in the Pell Cascade, and $D=0.82..$ in the Clapeyron Cascade (see Fig.~\ref{fig:superconvergences}).

Since the lengths of words grow exponentially fast within this class of cascades, the computational cost of `word lifting' becomes high very quickly. Moreover, the required precision increases as one searches for more data points. Due to the rapid convergence of $\lambda_n$ to a single value as $n$ increases, we require more and more decimal places to distinguish $\lambda_n$ from $\lambda_{n+1}$ for higher values of $n$. For that reason the number of data points is low, but still sufficient to illustrate the phenomenon. It is hard to claim universality of the second constant in Eq.~\eqref{convcontpar} based on the data; for example, in the Pell and Clapeyron cascades it is of order $\sim10^{-3}$, which is lower than the accuracy of our estimate of the parameter $D$. 

Other types of superconvergence have been reported before. For example, the values of $\lambda$ in the Shibayama cascade were shown to converge hypergeometrically, with~\cite{shibayama}
\begin{equation}
\Lambda_n/\Lambda_{n+1} \to 0.629... \quad \text{as } n \to \infty.
\end{equation}
However, as mentioned before, this cascade cannot be generated by substitution rules.

%
\section{Conclusions}
\label{sec:Conclusions}
%

\begin{figure}[t!]
\includegraphics[width=\linewidth]{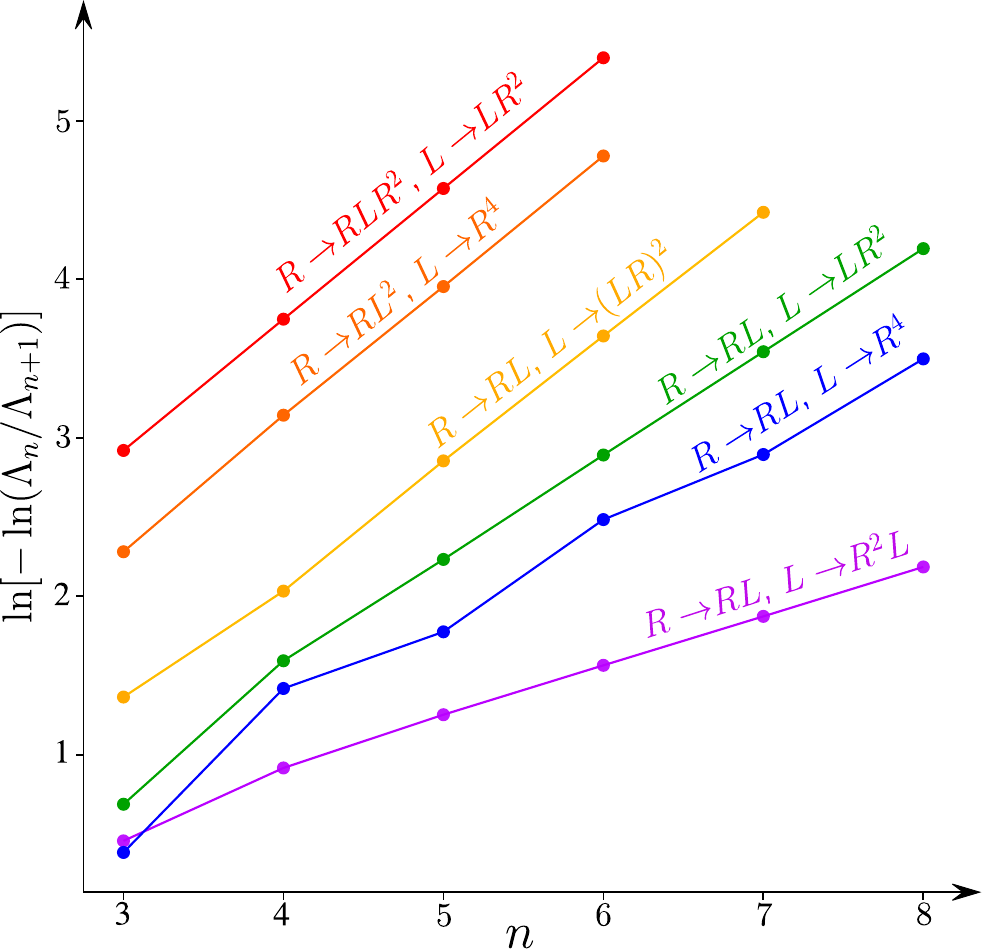}
\caption{The convergence of the control parameter $\lambda_n$ corresponding to the $n^{\textrm{th}}$ superstable orbit in a given cascade, measured by $\Lambda_n \equiv (\lambda_n-\lambda_{n+1})/(\lambda_{n+1}-\lambda_{n+2})$. The data were obtained using the `word lifting' technique applied to the logistic map $f_\lambda(x)=1-\lambda x^2$. The same procedure applied to the sine map $f_\lambda(x)=\lambda\sin(\pi x)$ produces plots indistinguishable by eye. The red and purple (top and bottom) cases correspond to the Clapeyron and Pell cascades respectively. }\label{fig:superconvergences}
\end{figure}

In this work we studied cascades of superstable orbits whose symbolic dynamics are generated by the repeated application of substitution rules. We demonstrated that, under certain general assumptions, substitution sequences can be rewritten as second-order linear recursive relations, allowing us to identify a double-exponential convergence of the topological entropy of successive words to their respective accumulation points. We found the convergence rates analytically for selected cascades. We numerically identified a quantitatively universal double-exponential convergence of the control parameters leading to successive words in the cascades. Finally, we identified a scaling property of the Devil's staircase of topological entropy: the widths of intervals within which all admissible words have lengths greater than an integer $n$ shrink geometrically as $n$ increases, in both topological entropy (with geometric ratio 2) and the control parameter (with a system-dependent geometric ratio, previously identified in Ref.~\onlinecite{PhysRevLett.47.975}). 

While we focussed on discrete-time unimodal maps in this work, the results hold more generally for maps in the same universality class. This class includes more physically-relevant continuous-time systems, such as the autonomous R\"{o}ssler attractor~\cite{Rossler76}, and periodically-driven systems such as the forced Brusselator~\cite{HaoEA83}. The wide applicability of the results again derives from the concept of universality: many physically relevant chaotic models feature sufficiently one-dimensional Poincar\'{e} first-return maps that their dynamics provide a good approximation to those considered here~\cite{badiipoliti, ChaosBook, anosovriemann,smalediff}. Driven dissipative systems will naturally tend to feature dynamics of this sort, as the dissipation causes a collapse onto a subset of the available phase space, but Hamiltonian systems can also demonstrate the same phenomena~\cite{AlligoodEA,ChaosBook}. A quantitative knowledge of the development of the complexity within such maps, as identified here, increases our understanding of the behaviour of the corresponding systems. The physical applications are wide-ranging, from chemical reactions, hydrodynamics, animal populations, and many more~\cite{RouxEA83,ArgoulEA87,GiglioEA81,mayr,strogatz:2000}.

A motivation for the present study was the observation that dissipative dynamical systems can spontaneously break the discrete time translation symmetry of a periodic driving by returning a response with the symmetries of a quasilattice~\cite{Fli18}. When stabilised to finite temperature by the local interactions of many degrees of freedom, the result can be termed a \emph{time quasicrystal}; several experiments were proposed and carried out in which time quasicrystals, or related phenomena, were reported~\cite{giergiel,PhysRevLett.109.163001,PhysRevA.97.012115,PhysRevLett.120.215301}. The Pell and Clapeyron cascades were previously shown to correspond to infinite sequences of successively-improving periodic approximations to time quasilattices, with the quasilattices themselves lying at the sequences' accumulation points~\cite{Fli18}. In the present work we extended this result by identifying time quasilattices in all ten physically-relevant quasilattice classes~\cite{BoyleSteinhardt16}. The results we presented here concerning the convergence of more general cascades including these cases as a subset, in particular the evolution of the ratios of successive control parameters, quantify the precision which would be required in any experimental investigation of time quasilattices via their periodic approximants.

We expect many of the results we have presented to generalise to the symbolic dynamics of systems divided into more than two partitions (described by alphabets containing correspondingly larger numbers of letters), since for discrete-time \emph{multimodal} maps in the logistic universality class the ordering of periodic windows within the chaotic regime changes, but the windows themselves remain~\cite{bailinhao}. Technical details make results in the multimodal case significantly harder to come by. While the monotonicity of topological entropy was proven for the logistic map in a series of early results~\cite{DH85,MThu,dMvS93,Dou95,Tsu00}, the extension to multimodal maps is only a recent development. Milnor's monotonicity of entropy theorem~\cite{Mil92,DGMT95} was proven for cubic maps in Ref.~\onlinecite{MT00}, and for general maps only very recently~\cite{Bruin2009MonotonicityOE,strienweixiao}.

Finding a link between the convergence of the control parameter $\lambda$ and the convergence of the topological entropy $h$ within discrete-time unimodal dynamical systems, such as that suggested in Ref.~\onlinecite{PhysRevE.51.1983}, could, in light of the work presented here, explain the quantitative form of the convergence of $\lambda$. This could point the way to a wider sense of universality, with applications in physics and other fields drawing on chaos theory.

%
\section{Acknowledgements}
The authors wish to thank Sebastian van Strien for helpful discussions. F.~F.~acknowledges support from the Astor Junior Research Fellowship of New College, Oxford.
%

\bibliographystyle{apsrev4-1}
\bibliography{mybib}

\end{document}